\documentclass[11pt,oneside,onecolumn,draftclsnofoot]{IEEEtran}
\usepackage[dvips]{graphicx}
\usepackage{epsfig}
\usepackage{color}
\usepackage{amsmath,amssymb,amsthm,amsfonts}
\usepackage{cite}
\def\xv{{\mathbf X}}
\def\xlv{{\underline{\mathbf X}}}
\def\xlt{{\tilde{\underline{\mathbf X}}}}
\def\yv{{\mathbf Y}}
\def\zv{{\mathbf Z}}

\def\hv{{\mathbf h}}

\def\Xv{{\mathbf X}}
\def\Yv{{\mathbf Y}}

\def\Wv{{\mathbf W}}

\def\Pc{{\mathcal{P}}}

\def\Psiv{{\mathbf \Psi}}

\def\vv{{\mathbf v}}

\def \Pr {{\mathrm{Pr}}}

\def\Cc{{\mathcal{C}}}

\def\Mc{{\mathcal{M}}}

\def\E{{\mathbb{E}}}

\def\Dd{{\mathcal D}}
\def\Du{{\mathcal U}}
\def \Wh {{\hat{W}}}
\def \Wt {{\check{W}}}
\def \rh {{\hat{r}}}
\def \rt {{\breve{r}}}
\def\Whv {{\hat{\mathbf W}}}
\def\Wtv {{\breve{\mathbf W}}}

\def\Bc{{\mathcal B}}
\def\Lc{{\mathcal L}}
\def\Tc{{\mathcal T}}

\def\Ee{{\mathcal {E}}}

\theoremstyle{definition}
\newtheorem{theorem}{Theorem}
\newtheorem{definition}{Definition}
\newtheorem{lemma}{Lemma}
\newtheorem{corollary}{Corollary}

\newtheorem*{lemma-A1}{Lemma A.1}
\IEEEoverridecommandlockouts

\begin{document}

\title{A Broadcast Approach To Secret Key Generation Over Slow Fading Channels}

\author{Xiaojun~Tang,~Ruoheng~Liu,~Predrag~Spasojevi\'{c}~and~H.~Vincent~Poor
\thanks{This research was supported by the National Science Foundation
under Grants CNS-09-05398, CCF-07-28208 and CCF-0729142, and in part by the Air Force Office of Scientific Research under Grant  FA9550-08-1-0480. The material in this paper was presented in part at the IEEE International Symposium on Information Theory (ISIT), Seoul, Korea, June 24 - 29, 2009.}%
\thanks{Xiaojun Tang is with AT\&T Labs, San Ramon, CA 94583 USA (e-mail:
xiaojun.tang@att.com). Ruoheng Liu is with Alcatel-Lucent, Murray Hill, NJ 07974 USA.  (email:
ruoheng.liu@alcatel-lucent.com). Predrag Spasojevi\'{c} is with the Wireless Information Network Laboratory (WINLAB), Department of Electrical and Computer Engineering, Rutgers University,
North Brunswick, NJ 08902, USA (e-mail: spasojev@winlab.rutgers.edu). H. Vincent Poor is with the Department of Electrical Engineering, Princeton University, Princeton, NJ 08544, USA (email: poor@princeton.edu).}}



\maketitle

\begin{abstract}
A secret-key generation scheme based on a layered broadcasting strategy is introduced for slow-fading channels. In the model considered, Alice wants to share a key with Bob while keeping the key secret from Eve, who is a passive eavesdropper. Both Alice-Bob and Alice-Eve channels are assumed to undergo slow fading, and perfect channel state information (CSI) is assumed to be known only at the receivers during the transmission.
In each fading slot, Alice broadcasts a continuum of coded layers and, hence,  allows Bob to decode at the rate  corresponding to the fading state (unknown to Alice). The index of a reliably decoded layer is sent back from Bob to Alice via a public and error-free channel and used to generate a common secret key.
In this paper, the achievable secrecy key rate is first derived for a given power distribution over coded layers. The optimal power distribution is then characterized. It is shown that layered broadcast coding can increase the secrecy key rate significantly compared to single-level coding.
\end{abstract}


\section{Introduction}\label{sec:intro}

Wireless secrecy has attracted considerable research interest due to the concern that wireless communication is highly vulnerable to security attacks, particularly eavesdropping attacks. Much recent research was motivated by Wyner's wire-tap channel model \cite{Wyner:BSTJ:75}, in which the transmission between two legitimate users (Alice and Bob) is eavesdropped upon by Eve via a degraded channel. In this model, to characterize the leakage of information to the eavesdropper, equivocation rate is used to denote the level of ignorance of the eavesdropper with respect to the confidential messages. Perfect secrecy requires that the equivocation rate is asymptotically equal to the message rate, and the maximal achievable rate with perfect secrecy is called the secrecy capacity. Wyner showed that secret communication is possible without a secret-key shared by legitimate users. Later, Csisz{\'{a}}r and K{\"{o}}rner generalized Wyner's model to consider general broadcast channels in \cite{Csiszar:IT:78}. The Gaussian wire-tap channel was considered in \cite{Leung-Yan-Cheong:IT:78}.  Recent research has addressed  the information-theoretic secrecy for multi-user channel models \cite{Liang:IT:06,Liu:ISIT:06,Lai:IT:06,Tekin:IT:07,Liu:IT:08,Liang:IT:09}.
We refer the reader to \cite{Liang:ITS:08} for a recent survey of the research progress in this area.

Interestingly, the wireless medium provides its own endowments that facilitate defending against eavesdropping. One such endowment is fading \cite{bigl:IT:98}.
The effect of fading on secret transmission has been studied in \cite{Gopala:IT:08, Liang:IT:08J, Li:Allerton:06}.
In these works, assuming that all communicating parties have perfect channel state information (CSI), the ergodic secrecy capacity has been derived. The scenario in which Alice has no CSI about Eve's channel (but knows the channel statistics) has also been studied in \cite{Gopala:IT:08}.  The throughput of several secure hybrid automatic repeat request (ARQ) protocols has been analyzed in \cite{Tang:IT:09}.  In this work, Alice is not assumed to have prior CSI (except channel statistics), but can receive a 1-bit ARQ feedback per channel coherence interval from Bob reliably.

Arguably, the most useful application of (keyless) secret message transmission is secret-key generation. For instance,  a key can be sent from Alice to Bob as a secret message (which is selected by Alice in advance). More generally, as considered here, the key can be established after a communication session completes. This relaxation in the protocol can lead to a higher key rate. The secret-key generation problem in \cite{Maurer:IT:93} and \cite{Ahlswede:IT:93PI} assumes an interactive, authenticated public channel with unlimited capacity.  In \cite{Ahlswede:IT:93PI},  the   ``channel model with wiretapper" (CW) is similar to the wiretap channel model, while in the ``source model with wiretapper" (SW), Alice and Bob exploit correlated source observations to generate the key. Both SW and CW models have
been subsequently extended to multiple terminals \cite{Csiszar:IT:00,Csiszar:IT:04,Csiszar:IT:08} and to non-authenticated public channels \cite{Maurer:IT:03, Maurer:IT:03B, Maurer:IT:03C}. Secret-key generation using both correlated sources and channels has been considered more recently in \cite{Khisti:ISIT:08} and \cite{Prabhakaran:ISIT:08}.

In this paper, we consider a key-generation problem in which Alice wants to share a key with Bob while keeping it secret from Eve. The Alice-Bob and Alice-Eve channels (forward channels) undergo slow fading, and CSI is known only at the receivers. Furthermore, we assume a public and error-free feedback channel. The key generation scheme under consideration consists of a communication and a key-generation phase. In the communication phase, via the forward channel, Alice sends to Bob coded sequences, which are observed at Bob and Eve after independent distortions due to power attenuation and noise. Subsequently, Alice and Bob agree on the same secret-key in the key-generation phase. The problem setting resembles an SW model but differs in that the shared ``correlated sources" are coded sequences (from a public codebook and distorted by the channel).
We assume that the feedback channel from Bob to Alice is very limited. For each block transmission from Alice to Bob, Bob is required to send back one or more bits to Alice, where the one-bit feedback corresponds to an ARQ ACK/NACK scheme. An example application is where Alice sends a video clip to Bob, which is a non-secret transmission. Bob responds with a few bits and thus enables agreeing on a secret-key, which can then be used in key-based cryptographic protocols.

The communication phase is based on layered broadcast coding, which effectively adapts the decoded rate at Bob to the actual channel state without requiring CSI to be available at Alice. The transmission takes place over several time
slots. In each time slot, Alice transmits a continuum of layers. Depending on the realization of the channel state, Bob decodes a subset of layers reliably. The index of the highest reliably decoded layer at Bob is sent back to Alice, and used in the key-generation phase that follows Wyner's secrecy binning scheme \cite{Wyner:BSTJ:75}. For a given power distribution over coded layers, we derive the achievable secrecy key rate, which permits a simple interpretation as the average reward collected from all possible channel realizations. Furthermore, we characterize the optimal power distribution over coded layers to maximize the achievable secrecy key rate under the broadcast approach.

Layered broadcast coding creates {\it artificial noise} so that the undecodable layers at Bob play the role of self-interference. We show that, by properly choosing the coding rate for each layer, it is ensured that Eve cannot  benefit from the layered coding structure and is forced to treat the layers undecodable at Bob as interference. Secret communications with interference was studied in \cite{Tang:ISIT:08} and \cite{Tang:arXiv:09A} in a more general (but non-fading) setting. Layered broadcast coding for a slow-fading single-input single-output (SISO) channel model was originally introduced by Shamai in \cite{Shamai:ISIT:97} and discussed in greater details in \cite{Shamai:IT:03}. The results in this paper are consistent with \cite{Shamai:ISIT:97} and \cite{Shamai:IT:03} when the additional secrecy key generation requirement phase is not considered. In a closely related work, a similar ARQ-based secret-key generation scheme employing single-level coding was studied in \cite{Ghany:ICC:09}. This scheme can be viewed as a special case of the proposed layered-coding based scheme as all power is allocated to a single coded layer. We show that layered broadcast coding can increase the secrecy key rate significantly compared to single-level coding.

The remainder of the paper is organized as follows. Section \ref{layer:model} describes the system model. Section \ref{layer:layered} states the broadcast approach for key generation. Section \ref{layer:srate} gives the achievable secrecy rate for a given power distribution. Section \ref{layer:power} characterizes the optimal power distribution. A numerical example involving a Rayleigh fading channel is given in Section \ref{layer:Rayleigh}. Conclusions are given in Section \ref{layer:conclusions}.

\section{System Model}\label{layer:model}

As depicted in Fig.~\ref{fig:model}, we consider a three-terminal model, in which Alice and Bob want to share a secret key in the presence of Eve, who is a passive eavesdropper. That is, Eve is interested in stealing the key but does not attempt to interfere with the key generation processes.

\subsection{Channel Model}\label{layer:model:channel}

\begin{figure*}
  \centering
  \includegraphics[width=0.7\linewidth]{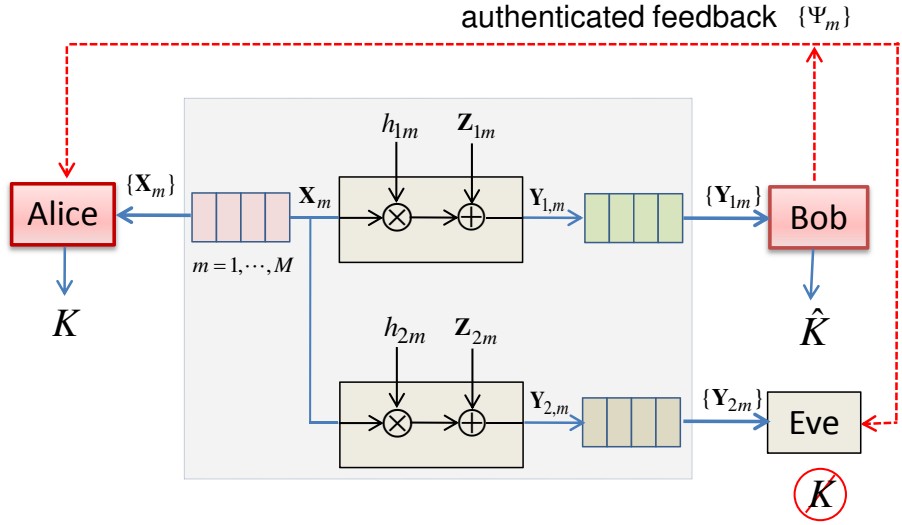}\\
  \caption{Alice and Bob want to agree on a key
  ($K=\hat{K}$), while keeping the key secret from Eve
  ($H(K|\Yv_2, \hv_2, \Psiv)/n \rightarrow 0$).}\label{fig:model}
\end{figure*}

The Alice-Bob and Alice-Eve channels (forward channels) undergo block fading, in which the channel gains are constant within a block while varying independently from block to block \cite{bigl:IT:98}. We assume that each block is associated with a time slot of duration $T$ and bandwidth $B$; that is, $N=\lfloor 2BT\rfloor$ real symbols can be sent in each slot. We also assume that the number of channel uses within each slot (i.e., $N$) is large enough to allow for invoking random coding arguments.

Let us assume that the transmissions in the forward channels take place over $M$ time slots. In a time slot indexed by $m \in [1,
\dots, M]$, Alice sends $\xv_m$, which is a vector of $N$ real symbols. Bob receives $\yv_{1m}$ through the channel gain $h_{1m}$ and Eve receives $\yv_{2m}$ through the channel gain $h_{2m}$. A
discrete time baseband-equivalent block-fading channel model can be expressed as
\begin{equation}\label{eq:chm}
  \yv_{tm} = \sqrt{h_{tm}} \xv_{m} + \zv_{tm}
\end{equation}
for $t=1,2$, where $\{\zv_{tm}\}$ are sequences of independent and identically distributed (i.i.d.) circularly symmetric complex Gaussian $\mathcal{N}(0, 1)$ random variables. We denote by $h_{1m}$ and $h_{2m}$ the states of the Alice-Bob and Alice-Eve channels, respectively, in time slot $m$. Without loss of generality, we drop the index $m$ and denote random channel realizations by $h_t$. We assume that $h_t$ is a real random variable with a probability density function (PDF) $f_t$ and a cumulative distribution function (CDF) $F_t$, for each $t=1,2$.
We also let $\hv_1=\left[h_{1,1},\dots, h_{1,M}\right]$ and $\hv_2=\left[h_{2,1},\dots, h_{2,M}\right]$ denote the power gain vectors for the Alice-Bob and Alice-Eve channels, respectively. We assume that Bob and Eve know their own channel gains perfectly; Alice does not know the CSI before its transmission, except for the channel statistics.

In addition, we assume a short term power constraint (excluding power variation across time slots) such that the average power of the signal $\xv_m$ per slot satisfies the constraint
\begin{equation}\label{layer-power}
    \frac{1}{N}E[\|\xv_m\|^2] \leq P
\end{equation}
for all $m=1,\dots,M$.

Finally, we assume that there exists an error-free feedback channel from Bob to Alice, through which Bob can feed back $\Psi_m$ for time slot $m$, where $\Psi_m$ is a deterministic function of $\Yv_{1m}$ and $h_{1,m}$. The feedback channel is assumed to be public, and therefore $\Psi_m$ is received by both Alice and Eve without any error.

\subsection{Secret Key Generation Protocol}\label{layer:model:protocol}

The secret key generation protocol consists of two phases: a communication phase and a key-generation phase.

\subsubsection{Communication Phase}\label{layer:model:protocol:comm}

We assume that the transmission during the communication phase takes place over $M$ time slots. That is, Alice sends a sequence of signals $\xv=(\xv_1, \xv_2, \dots, \xv_M)$ to the channel. Accordingly, Bob receives from his channel a sequence of signals denoted by $\yv_1=(\yv_{1,1}, \yv_{1,2}, \dots, \yv_{1,M})$ and Eve receives $\yv_2=(\yv_{2,1}, \yv_{2,2}, \dots, \yv_{2,M})$ from her channel. We let $n=MN$ denote the number of symbols sent by Alice in the communication phase.

After the transmission, Bob uses the feedback channel to send $\Psiv=(\Psi_1, \dots, \Psi_M)$, which is received by both Alice and Eve since the feedback channel is public and error-free.

\subsubsection{Key-Generation Phase}\label{layer:model:protocol:key-gen}

The communication phase is followed by a key-generation phase, in which both Alice and Bob generate the key based on the forward and backward signals. A general key-generation phase can be described as in the following.

Let $\mathcal{K}=\{1, 2, \dots, 2^{nR_s}\}$, where $R_s$ represents the secrecy key rate. Alice generates a secret key $k \in \mathcal{K}$ by using a decoding function $K$, i.e.,
\begin{equation}\label{alicekey}
    k=K\left(\xv,\Psiv\right).
\end{equation}
Bob generates the secret key $\hat{k} \in \mathcal{K}$ by using a decoding function $\hat{K}$, i.e.,
\begin{equation}\label{bobkey}
    \hat{k}=\hat{K}\left(\yv_1,\hv_1,\Psiv\right)= \hat{K}\left(\yv_1,\hv_1\right),
\end{equation}
where the second equality holds since we assume that $\Psiv$ is a deterministic function of $\yv_1$ and $\hv_1$.

The secrecy level at Eve is measured by the equivocation rate $R_e$ defined as the entropy rate of the key $K$ conditioned upon the observations at Eve, i.e.,
\begin{equation}\label{Eve-equivocation}
    R_e \triangleq \frac{1}{n}H(K|\Yv_2, \hv_2, \Psiv).
\end{equation}

\begin{definition}
A secrecy key rate $R_s$ is achievable if the conditions
\begin{align}
     &\Pr\left(K=\hat{K}\right) \geq 1-\epsilon, \label{reliable-key}\\
    \mbox{and} \qquad &R_e \geq R_s - \epsilon,\label{secret-key}
\end{align}
are satisfied for any $\epsilon>0$ as the number of channel uses $n \rightarrow \infty$.
\end{definition}

\section{A Layered Broadcast Approach To Key Generation}\label{layer:layered}

In this section, we introduce a broadcast approach for secret-key generation, in which Gaussian layered broadcast coding is used for the communication phase, and random secrecy binning is used for the key generation phase.

Before presenting the scheme, we briefly introduce Gaussian layered broadcast coding. Finite-level layered broadcast coding (superposition coding) was introduced by Cover in \cite{Cover:IT:72} for general broadcast channels. In \cite{Shamai:ISIT:97}, Shamai studied a Gaussian fading channel with no CSI at the transmitter and considered the limiting case when there is a continuum of coded layers. In this section, we first take a look at a fading wiretap channel with a finite number of fading states, for which finite level layered broadcast coding is applicable. The channel will be used to derive the result for the limiting case of continuous fading, which is the focus of this paper.

\subsection{Finite-Level Layered Broadcast Coding for $L$-State Fading Channel}\label{layer:layered:glc}

Let us first consider a type of channel called ``the $L$-state fading wiretap channel," in which there are $L$ different fading states possibly observed on the Alice-Bob or Alice-Eve channel.
\begin{figure*}
  \centering
  \includegraphics[width=0.8\linewidth]{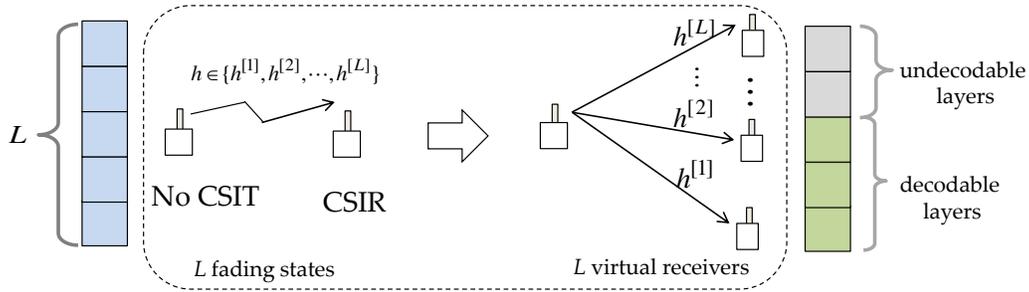}\\
  \caption{A point-to-point fading channel with $L$ possible fading states is viewed as a broadcast channel with $L$ virtual receivers each corresponding to a fading state.}\label{fig:layeredbroadcast}
\end{figure*}

\begin{definition}\label{def:L-state}
In an $L$-state fading wiretap channel, at any time slot, the realization of the  power gain of the Alice-Bob or Alice-Eve channel takes one value from $\{h^{[1]}, h^{[2]}, \dots, h^{[L]}\}$ independently and randomly, and is characterized by probability function $\Pr\{h_1=h^{[l_1]},h_2=h^{[l_2]}\}$. Without loss of generality, we assume that $\{h^{[l]}\}$ are ordered in ascending order.
\end{definition}

Here, let us focus on the Alice-Bob channel. As shown in Fig. \ref{fig:layeredbroadcast}, in a layered broadcast coding scheme, the point-to-point fading channel is viewed as a broadcast channel with $L$ virtual receivers each corresponding to a fading state. By applying the superposition coding in \cite{Cover:IT:72}, the encoding and decoding procedures can be described as follows.

During the encoding, we assume that $L$ layers are used. That is, the transmitted codeword is a superposition of $L$ codewords, i.e.,
  $\sum_{l=1}^{L} \xv^{[l]},$
 where $\xv^{[l]}$ is a codeword from a Gaussian codebook $\Cc^{[l]}$ with a rate $r^{[l]}$
and a constant power $p^{[l]}$, $l=1, \dots, L$. For a given power allocation $\{p^{[l]}\}$, the rate of the $l$-th layer is given by \footnote{All logarithms are to the natural base, and thus rates are in terms of nats per second per Hertz.}
\begin{equation}\label{ll-rate}
r^{[l]}=\log\left(1+\frac{h^{[l]}p^{[l]}}{1+h^{[l]}\sum_{i=l+1}^{L}p^{[i]}}\right),
\end{equation}
and the total power satisfies $\sum_{l=1}^{L} p^{[l]} = P$.

During the decoding, for a given fading realization $h^{[l]}$, the receiver can successfully decode the first $l$ layers by using the successive decoding strategy \cite{Cover:IT:72}.
i.e., the codewords $\{\xv^{[1]}, \dots, \xv^{[l]}\}$ can be decoded reliably,
while the codewords $\{\xv^{[l+1]}, \dots, \xv^{[L]}\}$ are undecodable. More specifically, in
the decoding process, the receiver first decodes $\xv^{[1]}$ by treating the
remaining codewords ($\{\xv^{[i]}, i>1$\}) as interference. After decoding
$\xv^{[1]}$, the receiver will subtract $\xv^{[1]}$ and then decode $\xv^{[2]}$
by treating the remaining codewords ($\{\xv^{[i]}, i>2$\}) as interference.
This process repeats until the $l$-th layer $\xv^{[l]}$ is decoded reliably by
treating the remaining codewords ($\{\xv^{[i]}, i>l$\}) as interference. As shown in (\ref{ll-rate}), $\sum_{i=l+1}^{L}p^{[i]}$ is the total power of coded layers treated as inference during the decoding of the $l$-th layer. Note that this predetermined ordering can be achieved because of the degraded nature
of Gaussian single-input single-output (SISO) channels.

\subsection{Layered Broadcast Coding for Gaussian Fading Channels}\label{layer:layered:glc}

In general, $L$ depends on the cardinality of the random channel variable. For a Gaussian fading channel, a continuum of code layers ($L \rightarrow \infty$) is required for achieving the best performance.
When a continuum of layers is used, the transmitter sends an
infinite number of layers of coded information. Each layer conveys
a fractional rate, denoted by $dR$, whose value depends on the
index of the layer. We refer to $s$, the realization of the fading
power, as a continuous index. For a given transmit power distribution $\rho(s)$ over coded layer $s$, $\rho(s)ds$ is the transmit power used by layer $s$. Any layer indexed by $u$ satisfying $u>s$ is undecodable and functions as additional interference. The total power of undecodable layers (for a realization of fading power $s$) is denoted by $I(s)$ and is expressed by
\begin{equation}\label{Interferences}
I(s)=\int_{s}^{\infty}\rho(u)du.
\end{equation}

The incremental differential rate of layer $s$ is
given by
\begin{equation}\label{dR}
  dR(s)=\log\left(1+\frac{s\rho(s)ds}{1+sI(s)}\right)=\frac{s\rho(s)ds}{1+sI(s)},
\end{equation}
where the second equality in (\ref{dR}) is due to the fact that $\lim_{x \rightarrow 0}\log(1+x)=x$ for any $x \geq 0$. The total power over all layers is constrained by
\begin{equation}\label{totalpower}
I(0)=\int_{0}^{\infty}\rho(u)du=P.
\end{equation}

Given a realization of the fading power (or layer index) $s$, the decodable rate at the receiver is
\begin{equation}\label{decodablerate}
    R(s) = \int_{0}^{s}\frac{u\rho(u)du}{1+uI(u)}.
\end{equation}
Hence, for a given CDF of the random fading power $s$ denoted by $F(s)$,  the average decodable rate at the receiver is
\begin{equation}\label{averagedecodablerate}
    R = \int_{0}^{\infty}\int_{0}^{s}\frac{u\rho(u)du}{1+uI(u)}dF(s).
\end{equation}

\subsection{Secret-Key Generation Based on Layered Broadcast Coding}\label{layer:scheme}

In this section, we discuss key generation based on Gaussian layered broadcast
coding. We outline the scheme for the continuous case when $L \rightarrow
\infty$, which is the focus of this paper. For an $L$-state fading wiretap channel when $L$ is finite, the corresponding scheme is discussed in Appendix \ref{layer-achievability}.

\subsubsection{Codebook Construction}\label{layer:scheme:codebook}

We need two types of codebooks used for the communication and key-generation phases, respectively.

The codebook used for the communication phase consists of a continuum of coded layers represented by $\{\Cc^{[s]}(2^{NdR(s)},N)\}$, where $N$ is the codeword length and $dR(s)$ is the (incremental differential) rate at layer $s$. The (sub-)codebook for each layer is generated randomly and independently. That is, for any codebook $\Cc^{[s]}(2^{NdR(s)},N)$, we generate $2^{NdR(s)}$ codewords $\xv^{[s]}(w)$, where $w=1,2,\dots, 2^{NdR(s)}$, by choosing the $N{2^{NdR(s)}}$ Gaussian symbols (with power $\rho(s)ds$) independently at random.

The codebook used for the key generation phase is based on Wyner's secrecy coding \cite{Wyner:BSTJ:75, Gopala:IT:08}. As shown in Fig. \ref{fig:keybook}, we use
\begin{equation}\label{regularrate}
    R = \int_{0}^{\infty}\int_{0}^{h_1}\frac{s\rho(s)ds}{1+sI(s)}dF_1(h_1)
\end{equation}
to represent the average decodable rate at Bob. We first generate all binary sequences of length $n(R-\epsilon)$, denoted by $\Bc$, where $n=MN$. The sequences $\Bc$ are then randomly and uniformly grouped
into $K=2^{nR_s}$ bins each with $n(R-R_s-\epsilon)$ sequences, where $R_s$ is the achievable secrecy rate given later. We denote by $\vv(k,j)$ the $j$-th codeword in the $k$-th bin, where $1 \leq k \leq K$ and $1 \leq j \leq J=2^{n(R-R_s-\epsilon)}$. Each secret key $k \in \{1, \dots, K\}$ is then randomly assigned to a bin, denoted by
$\Bc(k)=\{\vv(k,j), j=1,\dots, J\}$.

\subsubsection{Communication Phase}\label{layer:scheme:comm}

The communication takes places over $M$ time slots. In time slot $m \in {[1,\dots,M]}$, Alice first randomly selects a message $W_m^{[s]}\in \{1,\dots, 2^{NdR(s)}\}$ for coded layer $s$, independent of the message chosen for other layers. For convenience, we use $W_m$ to represent the total message sent in time slot $m$ (through all layers), i.e., $W_m = \times_{s} W_m^{[s]}$. Then, Alice sends a superposition of all layers to the channel.

Bob receives $\yv_{1m}$ and tries to decode all his decodable layers, which depends on his channel state $h_{1m}$. For convenience, we use $W_m^{[\Dd_1]}$ to denote the set of layers reliably decoded by Bob, and $W_m^{[\Du_1]}$ to denote the set of layers undecodable to Bob in time slot $m$.\footnote{To be more accurate, $\Dd_1$ in $W_m^{[\Dd_1]}$ should be indexed by $m$, however, we choose to use $\Dd_1$ to simplify our notation. Throughout the paper, $W_m^{[\Dd_1]}$ is shorthand for $W_m^{[\Dd_{1m}]}$. If the subscript of $W$ is a set, then $\Dd_1$ is also indexed by the set. For example, for a set of time slots $\Mc^{+} \subseteq \{1,\dots,M\}$, we use $W_{\Mc^{+}}^{[\Dd_{1}]}$ instead of $W_{\Mc^{+}}^{[\Dd_{1{\Mc^{+}}}]}$ to represent all the messages decoded by Bob in $\Mc^{+}$. The rule is also applied to $\Dd_2$, $\Du_1$ and $\Du_2$. In addition, it is applied to codeword $\Xv$ and codebook $\Cc$ besides message $W$.}  After decoding, Bob sends back the index of the highest decodable layer to Alice via the feedback channel, so that both Alice and Bob get to know $W_m$.  This completes the transmission in time slot $m$. The communication phase ends when all $M$ (independent) transmissions are completed.

Note that the feedback of a layer index does not need to be completed right after each transmission in the forward channel. It is required only before the following key generation phase. Also note that the feedback of the index of a decodable layer is a special type of channel feedback. In particular, when considering the case when the number of fading states $L \rightarrow \infty$, the index of the highest decodable layer in time slot $m$ is equal to the fading power gain $h_{1m}$ (i.e., the public feedback $\Psi_m=h_{1m}$). For a finite level layered coding approach, the feedback of the layer index is an $L$-bit quantized version of the realization of the fading power gain. When $L=1$, it is the ARQ feedback of ACK or NACK.

\subsubsection{Key-Generation Phase}\label{layer:scheme:keygen}

\begin{figure*}
  \centering
  \includegraphics[width=0.85\linewidth]{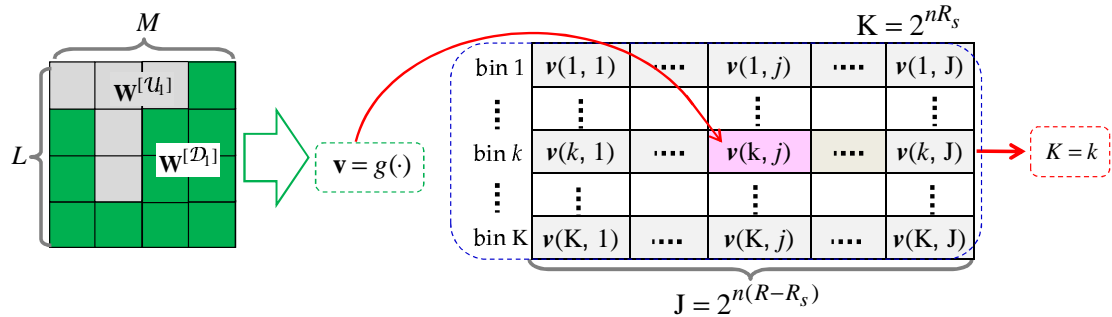}\\
  \caption{Alice and Bob generate a sequence $\vv$ from all the messages reliably decoded (across $L$
layers and $M$ time slots), look up in the key-generation codebook for a $k$ such that $\vv \in \Bc(k)$, and output $k$ as the key.}\label{fig:keybook}
\end{figure*}

Once the communication phase (including feedback) is completed, both Alice and Bob can generate the
secret key. Based on the feedback sequence $\Psiv = \hv_1$,  Alice generates a
binary sequence $\vv$ from all the messages reliably decoded by Bob based on any deterministic one-to-one mapping $g$ as
\begin{equation}\label{one-2-one}
    \vv=g(\Wv^{[\Dd_1]}),
\end{equation}
where $\Wv^{[\Dd_1]} = (W_1^{[\Dd_1]}, W_2^{[\Dd_1]}, \dots, W_2^{[\Dd_1]})$ represents the set of messages successfully decoded by Bob across all
layers and time slots.

Alice then looks up in the key-generation codebook for a $k$ such that $\vv \in \Bc(k)$, and outputs $k$ as the secret key generated.
Note that all those messages are decoded by Bob, and Bob can generate the same sequence $\vv$ and the same key $k$ as Alice does. This completes the key generation.

\section{Secrecy Key Rate}\label{layer:srate}

In this section, we present the secrecy key rate achieved by the broadcast approach and compare it to that achieved by using a single-level coding approach. For both approaches, we assume that the number of time slots used in the transmission over the forward channel is sufficiently large (i.e., $M \rightarrow \infty$), so that we can obtain an ergodic key rate.

\subsection{Layered-Broadcast-Coding Based Key Generation}

The following result characterizes the secrecy rate when a power distribution $\rho(s)$ is given.

\begin{theorem}\label{th:srate_cnt}
For a given power distribution $\rho(s)$ over coded layers indexed by $s$, the secrecy key rate achieved by the layered-broadcast-coding based key generation scheme is
\begin{equation}\label{srate_cnt}
    R_s=\int_{0}^{\infty}\int_{0}^{h_1}\Delta(h_1,h_2)dF_2(h_2)dF_1(h_1),
\end{equation}
where $\Delta(h_1,h_2)$ is given by
\begin{equation}\label{deltas}
\Delta(h_1,h_2)=\int_{h_2}^{h_1}\left[\frac{s\rho(s)}{1+sI(s)}-\frac{h_2\rho(s)}{1+h_2I(s)}\right]ds
\end{equation}
and
\begin{equation}\label{Rlterm}
    I(s)= \int_{s}^{\infty}\rho(u)du \mbox{~~ with~$I(0)=P$}.
\end{equation}
\end{theorem}
\begin{proof}
The proof can be found in Appendix \ref{layer-achievability}.
\end{proof}

Now we discuss some insights from Theorem \ref{th:srate_cnt}. First, $R_s$ can be written as
\begin{equation}\label{srate_cntv}
  R_s=\E_{h_1,h_2}\left[\tilde{\Delta}(h_1,h_2)\right],
\end{equation}
where
\begin{equation}\label{Delta2}
\tilde{\Delta}(h_1,h_2) = \left\{ \begin{array}{ll}
 \Delta(h_1,h_2) &\mbox{ if $h_1 > h_2$} \\
  0 &\mbox{ otherwise.}
  \end{array} \right.
\end{equation}
The key rate $R_s$ is the average of rewards (designated by
$\tilde{\Delta}(h_1,h_2)$) collected from all possible channel realizations.
Positive rewards are obtained from the time slots in which Bob's channel is
better than Eve's channel ($h_1 > h_2$). On the other hand, when $h_1 \leq
h_2$, the reward is zero.

We can see that except for the rare case in which $h_1$ is always smaller than $h_2$, $R_s$ is positive.

Now we focus on a particular time slot $m$ in which $h_1>h_2$, and use $\xv_m$ to denote all
layers sent in the slot. \footnote{$\xv_m$ represents the set of $L$ layers in time slot $m$, and also the signal transmitted by Alice in time slot $m$, which is the superposition of all layers.}  As depicted in Fig. \ref{fig:layering}, $\xv_m$ can be
divided as
\begin{equation}\label{layerset}
\xv_m=\xv_m^{[\Dd_2]} \cup \left(\xv_m^{[\Dd_1]} \cap \xv_m^{[\Du_2]} \right) \cup \xv_m^{[\Du_1]},
\end{equation}
where $\xv_m^{[\Dd_1]}$ and $\xv_m^{[\Du_1]}$ denote the sets of decodable and
undecodable layers at Bob, respectively, and $\xv_m^{[\Dd_2]}$ and
$\xv_m^{[\Du_2]}$ denote the sets of decodable and undecodable layers at
Eve, respectively. Note that $\xv_m^{[\Dd_1]} \supset \xv_m^{[\Dd_2]}$ since
$h_1>h_2$.

Both Alice and Bob can decode $\xv_m^{[\Dd_2]}$, and neither of them can decode $\xv_m^{[\Du_1]}$. Therefore, a nonzero reward $\Delta(h_1,h_2)$ comes from the set of layers $\xv_m^{[\Dd_1]} \cap \xv_m^{[\Du_2]}$. To show this, we rewrite (\ref{deltas}) as
\begin{equation}\label{deltas_f2}
\Delta(h_1,h_2)=\int_{h_2}^{h_1}\frac{s\rho(s)ds}{1+sI(s)}-\int_{h_2}^{h_1}\frac{h_2\rho(s)ds}{1+h_2I(s)}.
\end{equation}
The first term on the right hand side of (\ref{deltas_f2}) is the sum-rate decoded by Bob from $\xv_m^{[\Dd_1]} \cap \xv_m^{[\Du_2]}$ (by decoding and canceling $\xv_m^{[\Dd_2]}$ first, and treating the interference term $\xv_m^{[\Du_1]}$ as noise). Furthermore, the second term can be written as
\begin{equation}\label{deltas_f2_t2}
\int_{h_2}^{h_1}\frac{h_2\rho(s)ds}{1+h_2I(s)}=\log\left(1+\frac{h_2\left[I(h_2)-I(h_1)\right]}{1+h_2I(h_1)}\right).
\end{equation}
By noticing that $I(h_2)-I(h_1)$ is the total power used for the layers $\xv_m^{[\Dd_1]} \cap \xv_m^{[\Du_2]}$, and $I(h_1)$ is the total power used for the layers $\xv_m^{[\Du_1]}$, (\ref{deltas_f2_t2}) gives the rate of information that Eve can possibly deduce from $\xv_m^{[\Dd_1]} \cap \xv_m^{[\Du_2]}$ through her channel with power gain $h_2$.

An interesting finding here is that what the best Eve can do is to treat the interference term $\xv_m^{[\Du_1]}$ as noise (as Bob does) with the total noise power $1+h_2I(h_1)$, and therefore cannot benefit from the structure of interference either.
Due to the absence of CSI at the transmitter during the transmission in the forward channel , the
layered broadcast coding strategy creates a medium with interference, in which the
undecodable layers play the role of \emph{self-interference}. We remark that this is a special case of secret communication over a medium with interference as discussed in \cite{Tang:arXiv:09A}.

\begin{figure}[t]
  \centering
  \includegraphics[width=0.6\linewidth]{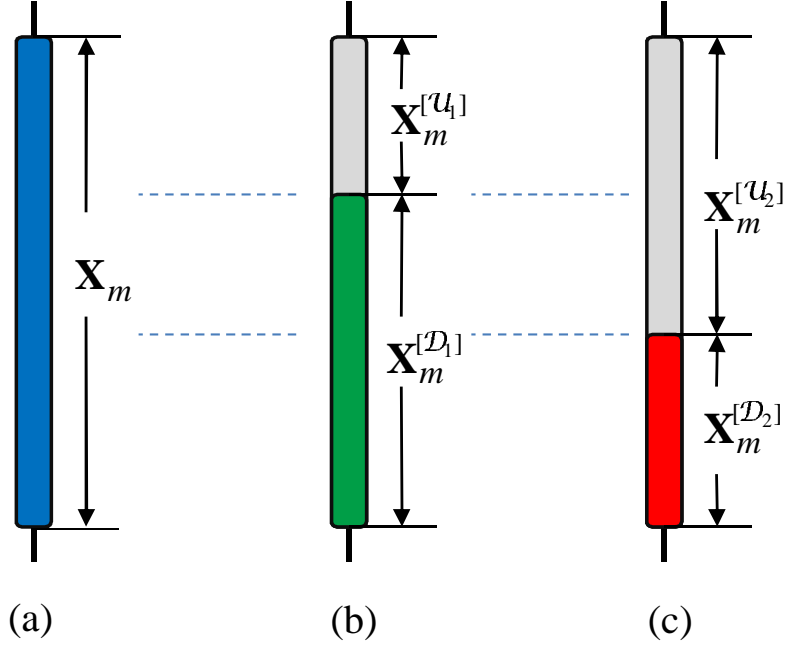}\\
  \caption{(a) Coded layers sent by Alice, (b) decodable and undecodable layers for Bob, and
  (c) decodable and undecodable layers for Eve, in time slot $m$ with the channel gains $h_1>h_2$. }
  \label{fig:layering}
\end{figure}

\subsection{Single-Level-Coding Based Key Generation}

When single-level coding is used, self-interference does not occur. Alice uses a codebook with a single coding rate in the forward transmission. Bob uses ARQ feedback to tell Alice whether the decoding is successful or has failed.  In this case, the following secrecy key rate can be achieved.
\begin{lemma}\cite[Theorem~$1$]{Ghany:ICC:09}\label{single-level-lemma}
The secrecy key rate of a single-level-coding based scheme is given by
\begin{equation}\label{srate-onelevel}
    R_{s}^{[1]}=\Pr\left[R^{[1]} \leq \log(1+h_1P)\right]\E_{h_2}\left[R^{[1]}- \log\left(1+h_2P\right)\right]^{+},
\end{equation}
where $R^{[1]}$ is the coding rate of the single-level codebook.
\end{lemma}

This key rate $R_{s}^{[1]}$ still has the interpretation of the average of rewards (designated by
$\tilde{\Delta}_1(h_1,h_2)$) collected from all possible channel realizations. That is, $R_{s}^{[1]}$ can be written as
\begin{equation}\label{srate_cntv_1layer}
  R_{s}^{[1]}=\E_{h_1,h_2}\left[\tilde{\Delta}_1(h_1,h_2)\right],
\end{equation}
where
\begin{equation}\label{Delta2_1layer}
\tilde{\Delta}_1(h_1,h_2) = \left\{ \begin{array}{ll}
 R^{[1]}-\log(1+h_2P) &\mbox{ if $h_1 \geq \frac{\exp(R^{[1]})-1}{P} > h_2$} \\
  0 &\mbox{ otherwise.}
  \end{array} \right.
\end{equation}

\subsection{Comparisons and Discussions}

The advantage of the layered-broadcast-coding (LBC) based approach over the single-level-coding based approach (SLC) can be readily observed by comparing the reward functions given by (\ref{Delta2}) and (\ref{Delta2_1layer}). First, in LBC, a positive reward is obtained from the set of channel pairs $\Pc=\{(h_1,h_2):h_1>h_2\}$; while in SLC, it is obtained from the channel set $\Pc'=\{(h_1,h_2):h_1\geq \frac{1}{P}(e^{R^{[1]}}-1) > h_2\}$. It is obvious that $\Pc \supset \Pc'$, which means there are more time slots that contribute to the secrecy key generation for LBC than for SLC. Second, the coding rate $R^{[1]}$ for SLC has to be carefully chosen in order to balance between obtaining a larger value of reward in a time slot (by increasing $R^{[1]}$) and making more time slots contribute to the key generation (by decreasing $R^{[1]}$); while in LBC, the reward is gained in each time slot adaptively based on the random channel realizations. Finally and importantly, in SLC, Eve can deduce the information at the rate of $\log(1+h_2P)$ with a channel gain $h_2$. This is the loss of rate in order to keep the key secret from Eve. In LBC, however, Eve deduces less information as given by (\ref{deltas_f2_t2}) due to the interference power (the total power of undecodable layers). The self-interference plays an important role for decreasing Eve's capability of eavesdropping.

Hence, although the single-level-coding based approach has lower decoding complexity, and requires less feedback (only 1-bit per time slot), it is sub-optimal in general (when feedback of multiple bits is allowed).
By all means, the single-level coding scheme can be considered as a special case of a layered-broadcast-coding based scheme, in which all power is allocated to a single layer. It serves as a baseline scheme and further motivates us to find the best power distribution for optimizing the layered-broadcast-coding scheme.

\section{Optimal Power Distribution}\label{layer:power}

In this section, we derive the optimal distribution of power over coded layers for our broadcast approach.
The secrecy rate given by (\ref{srate_cnt}) is hard to evaluate and optimize due to the three-dimensional integrals. After some steps of derivations, we have an alternative form given as follows:

\begin{lemma}\label{lm:srate2_cnt}
The secrecy key rate given by (\ref{srate_cnt}) is equivalent to
\begin{equation}\label{srate2_cnt}
    R_s=\max_{I(x)} \int_{0}^{\infty}\left[1-F_1(x)\right]\rho(x)\left[\int_{0}^{x}\frac{F_2(y)dy}{[1+yI(x)]^{2}}\right]dx,
\end{equation}
with the constraint $I(0)=P$, and $\rho(x)=-dI(x)/dx$.
\end{lemma}
\begin{proof}
The proof can be found in Appendix \ref{srate2}.
\end{proof}

\subsection{Optimal Interference Distribution}

In certain cases, optimization of $R_s$ with respect to the power distribution $\rho(x)$, or equivalently, the interference distribution $I(x)$, under the power constraint $P$ can be found by using the calculus of variations. First, we define the functional of (\ref{srate2_cnt}) as
\begin{equation*}\label{functional}
    L\left(x,I(x),I'(x)\right)=-\left[1-F_1(x)\right]I'(x)\left[\int_{0}^{x}\frac{F_2(y)dy}{[1+yI(x)]^{2}}\right].
\end{equation*}
A necessary condition for a maximum of the integral of $L(x,I(x),I'(x))$ over
$x$ is a zero variation of the functional. By solving the associated
E\"{u}ler-Lagrangian equation \cite{Gelfand:63} given as
\begin{equation}\label{euler}
    \frac{\partial L}{\partial I}-\frac{d}{dx}\left(\frac{\partial L}{\partial I'}\right)=0,
\end{equation}
we have the following characterization for the optimal $I(x)$.

\begin{theorem}\label{lm:Iterm_cnt}
A necessary condition for optimizing $I(x)$ in order to maximize the secrecy rate given by (\ref{srate2_cnt}) is to choose $I(x)$ to satisfy
\begin{equation}\label{i_cnt}
\int_{0}^{x}\frac{F_2(y)dy}{[1+yI(x)]^2}=\frac{\left[1-F_1(x)\right]F_2(x)}{f_1(x)\left[1+xI(x)\right]^2},
\end{equation}
where $I(x)=0$ when $x < x_0$ or $x \geq x_1$. Here, $x_0$ and $x_1$ can be found by setting $I(x_0)=P$ and $I(x_1)=0$ in (\ref{i_cnt}).
\end{theorem}

\begin{proof}
The proof can be found in Appendix \ref{Iterm}.
\end{proof}
\vspace{0.1in}

In general, numerical computation is needed for solving (\ref{i_cnt}) in order to obtain the optimal interference distribution $I(x)$. For some special CDFs $F_2(x)$, an analytical form of $I(x)$ is possible if the integral in (\ref{i_cnt}) can be evaluated in a closed form.

In the following, we consider two of such special cases:

\subsubsection{Non-Fading Alice-Eve Channel}

If the Alice-Eve channel is constant with channel power gain $x^{*}$, the CDF $F_2(x)$ is $F_2(x)=\mu(x-x^{*})$, where $\mu(x)$ represents a unit step function. In this case, the optimal interference distribution is given by
\begin{equation}\label{Evenonfading}
I(x)=\frac{1-F_1(x)-(x-x^{*})f_1(x)}{x(x-x^{*})f_1(x)-x^{*}\left[1-F_1(x)\right]},
\end{equation}
which can be easily shown from (\ref{i_cnt}).

\subsubsection{Non-Secret Layered Transmission}

If key-generation is not considered and it is desired to find the optimal $I(x)$ to maximize the average reliably
decodable rate at Bob in the non-secret layered transmission, this can be done by assuming $x^{*}=0$ in (\ref{Evenonfading}).
In this case, we have
\begin{equation}\label{Nonseclayer}
I(x)=\frac{1-F_1(x)}{x^2f_1(x)}-\frac{1}{x},
\end{equation}
which is consistent with the result given in \cite{Shamai:IT:03}.

\subsection{Secrecy Key Rate With Optimal Power Distribution}

Finally, we have the following secrecy key rate under the optimal power distribution.
\begin{corollary}\label{th:srate2}
When the optimal power distribution is used, the following secrecy key rate is achieved:
\begin{equation}\label{srate_cnt3}
    R_s=\int_{x_0}^{x_1}\frac{-\left[1-F_1(x)\right]^2F_2(x)dI(x)}{f_1(x)[1+xI(x)]^{2}},
\end{equation}
where $I(x)$ and $(x_0,x_1)$ are found from the condition given by Theorem \ref{lm:Iterm_cnt}.
\end{corollary}
\begin{proof}
The proof is straightforward by combining Lemma \ref{lm:srate2_cnt} and Theorem \ref{lm:Iterm_cnt}.
\end{proof}

\section{A Rayleigh Fading Channel}\label{layer:Rayleigh}

In this section, we assume Rayleigh fading for both Alive-Bob and Alice-Eve channels. The fading gains $h_t$ are exponentially distributed with means $\lambda_t$ for $t=1,2$. That is, the PDFs of the fading gain $h_t$ are \begin{equation}\label{pdf_cnt}
    f_t(s) = \left\{ \begin{array}{ll}
 \frac{1}{\lambda_t}\exp\left(-\frac{s}{\lambda_t}\right) &\mbox{ if $s \geq 0$,} \\
 0  &\mbox{otherwise,}
       \end{array} \right.
\end{equation}
for $t=1,2$ and the CDFs are
\begin{equation}\label{cdf_cnt}
    F_t(s) = \left\{ \begin{array}{ll}
 1-\exp\left(-\frac{s}{\lambda_t}\right) &\mbox{ if $s \geq 0$,} \\
 0  &\mbox{otherwise.}
       \end{array} \right.
\end{equation}

\subsection{Single-Level-Coding Approach}

For comparison, we first calculate the secrecy key rate when single-level coding is used. As shown in Appendix \ref{proof_onelevel_term1}, the secrecy rate is
\begin{align}\label{rate_l1}
    R_s^{[1]}=\max_{R^{[1]} \geq 0} &\quad \exp\left(-\frac{e^{R^{[1]}}-1}{\lambda_1 P}\right)\left\{R^{[1]} -\exp\left(\frac{1}{\lambda_2 P}\right) \left[E_i\left(\frac{e^{R^{[1]}}}{\lambda_2 P}\right)-E_i\left(\frac{1}{\lambda_2 P}\right)\right]\right\},
\end{align}
where $E_i(x)=\int_{x}^{\infty}[\exp(-t)/t]dt$ is the exponential integral function. It can be verified that the above function is concave with respect to $R^{[1]}$ and thus has a unique maximum, which can be searched numerically.

\subsection{Layered-Coding Approach}

According to (\ref{srate_cnt3}), the secrecy rate with layered coding under the optimal power control is computed  numerically by evaluating
\begin{equation*}\label{rate_linf}
    R_s = \lambda_1 \int_{x_0}^{x_1}\frac{\exp(-x/\lambda_1)\left[\exp(-x/\lambda_2)-1\right]}{[1+xI(x)]^2}dI(x),
\end{equation*}
where the optimal interference distribution $I(x)$ and boundary points $x_0$ and $x_1$ can be found according to Lemma \ref{lm:Iterm_cnt} as follows.

\textit{Interference Distribution $I(x)$}

As shown in Appendix \ref{proof_int_term1}, we have
\begin{equation}\label{int_term_ap}
   \int_{0}^{x}\frac{F_2(y)dy}{\left[1+yI(x)\right]^2}=\frac{\exp\left(-x/\lambda_2\right)-1}{I(x)\left[1+xI(x)\right]}+\frac{\exp\left({1}/{\lambda_2I(x)}\right)}{\lambda_2I^2(x)}\left[E_i\left(\frac{1}{\lambda_2I(x)}\right)-E_i\left(\frac{1+xI(x)}{\lambda_2I(x)}\right)\right].
\end{equation}
We also have
\begin{equation}\label{int_term2}
   \frac{\left[1-F_1(x)\right]F_2(x)}{f_1(x)\left[1+xI(x)\right]^2}=\frac{\lambda_1\left[1-\exp(-x/\lambda_2)\right]}{\left[1+xI(x)\right]^2}.
\end{equation}

Therefore, we can show after some steps of arrangements that $I(x)$ is found by solving
\begin{align}\label{Iequation}
    & E_i\left(\frac{1}{\lambda_2I(x)}\right)-E_i\left(\frac{1+xI(x)}{\lambda_2I(x)}\right) \\
    & = \frac{\lambda_2 I(x)[1+\lambda_1 I(x)]}{[1+xI(x)]^2}\left[\exp\left(-\frac{1}{\lambda_2I(x)}\right)-\exp\left(-\frac{1+xI(x)}{\lambda_2 I(x)}\right)\right].\notag
\end{align}

\textit{Boundary Points $x_0$ and $x_1$}

We needs to find the boundary points $x_0$ and $x_1$ to meet the constraints that
\begin{equation*}
I(x_0)=P \mbox{~~and~~} I(x_1)=0.
\end{equation*}

By letting $I(x_0)=P$ in (\ref{Iequation}), we can solve the equation for $x_0$. However, $x_1$ cannot be solved by this means since we cannot let $I(x_1)=0$ in (\ref{Iequation}). Instead, we let $I(x_1)=0$ in (\ref{i_cnt}) and find that
\begin{equation*}
\int_{0}^{x_1}F_2(y)dy = x_1 + \lambda_2\left[\exp(-x_1/\lambda_2)-1\right],
\end{equation*}
and
\begin{equation*}
\frac{\left[1-F_1(x_1)\right]F_2(x_1)}{f_1(x_1)}=\lambda_1\left[1-\exp\left(-x_1/\lambda_2\right)\right].
\end{equation*}
Therefore, $x_1$ can be found by solving the following equation:
\begin{equation*}
x_1+(\lambda_1+\lambda_2)\left[\exp\left(-\frac{x_1}{\lambda_2}\right)-1\right]=0.
\end{equation*}

Interestingly, $x_1$ depends only on the channel statistics (characterized by $\lambda_1$ and $\lambda_2$ for the Rayleigh fading channels) and not on the power constraint $P$. Note that no power will be allocated to a layer with its index higher than $x_1$ (however, it is possible that some layers lower than $x_1$ still have zero power allocation, as shown in the numerical example). Finally, we remark that every equation discussed in this section has a unique solution after excluding a trivial solution $0$.

\subsection{Numerical Examples}

Now we show some numerical examples on the achievable secrecy-key rates and the optimal power distribution $\rho(s)$. We consider the symmetric Rayleigh fading channel defined by (\ref{pdf_cnt}) with $\lambda_1=\lambda_2=1$.

Fig. \ref{fig:keyrate} compares the secrecy key rates achieved by the layered-coding and single-level-coding based schemes (both optimized). We also compare them with the secrecy rate when perfect and noncausal CSI of the Alice-Bob channel is available to Alice. In this case, Alice is able to adapt its transmission rate based on the CSI at each time slot. We still assume a short-term power constraint and thus Alice does not adapt power in contrast to the scheme given by \cite{Gopala:IT:08}. Without CSI at Alice, the secrecy key rate achieved by the layered-coding based scheme is significantly higher. This shows the benefit of the broadcast approach due to the introduction of self-interference in transmission.
\begin{figure*}
  \centering
  \includegraphics[width=0.8\linewidth]{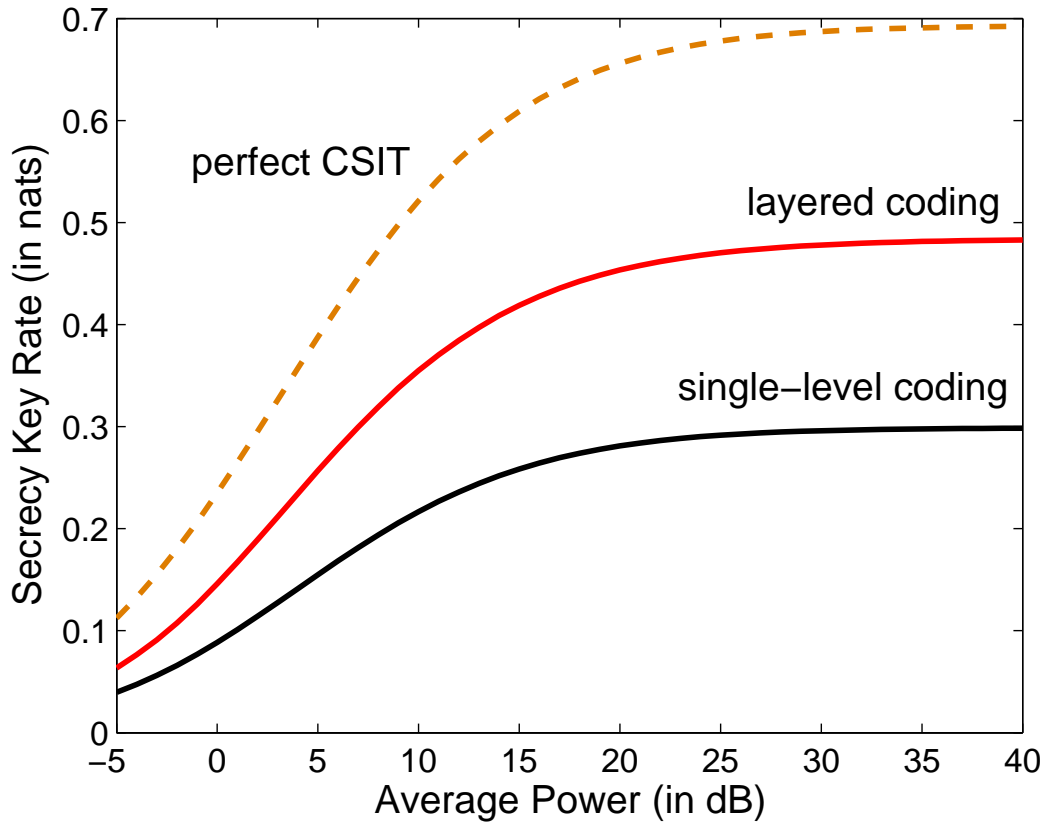}\\
  \caption{Secrecy key rates achievable for the layered-coding-based approach, the single-level-coding-based approach, and when perfect CSIT is available at Alice noncausually.}\label{fig:keyrate}
\end{figure*}

Fig. \ref{fig:power} shows the optimal power distribution over coded layers. A trend is that more power is distributed to lower layers as the total transmit power $P$ becomes larger. In general, the optimal power distribution does not concentrate much on a certain layer (or a small set of layers), especially when $P$ is large. We also compare the optimal power distribution for maximizing the secrecy key rate in key-generation and that for maximizing the average reliably decodable rate at Bob in non-secret transmission. With different power constraints, the power distributions for non-secret transmission are on the same curve but have different boundary points, which is different from the case for key generation. Also, when the total transmit power exceeds a certain threshold, the power distribution for key generation is more concentrated over higher layers (as shown for the cases of $P=5$ and $P=20$); while the opposite can be observed when $P$ is small (as shown for the case of $P=1$ in Fig. \ref{fig:power}.)

\begin{figure*}
  \centering
  \includegraphics[width=0.8\linewidth]{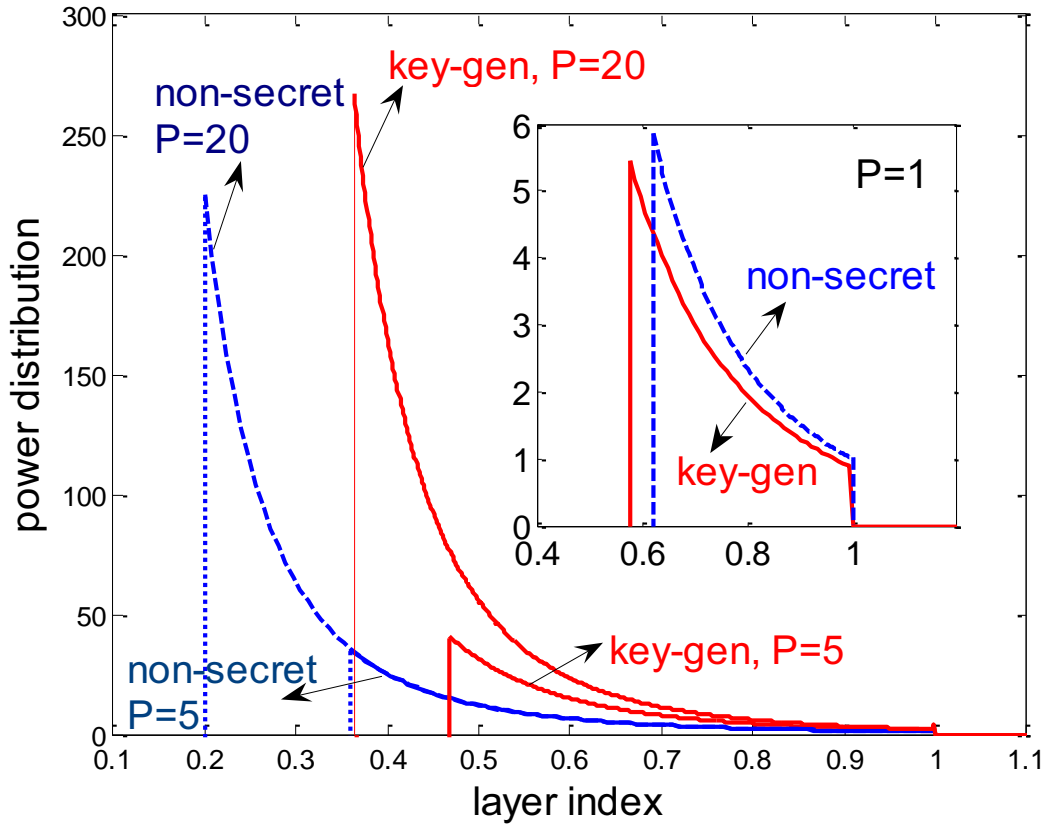}\\
  \caption{Optimal power distributions for maximizing the secrecy key rate in key-generation (``key-gen") and for maximizing the average reliably decodable rate at Bob in non-secret transmission (``non-secret") when the normalized transmit power is $P=1, 5, 20$. }\label{fig:power}
\end{figure*}

\section{Conclusions}
\label{layer:conclusions}

In this paper, we have introduced a broadcast approach for secret-key generation over slow-fading channels based on layered broadcast coding. We have considered a model in which Alice attempts to share a key with Bob while keeping the key secret from Eve. Both Alice-Bob and Alice-Eve channels are assumed to undergo slow fading, and perfect CSI is assumed to be known only at the receivers during the transmission. Layered coding facilitates adapting the reliably decoded rate at Bob to the actual channel state without CSI available at Alice. The index of a reliably decoded layer is sent back to Alice via an authenticated, public and error-free channel, which is exploited by Alice and Bob to generate the secret key. We have derived the achievable secrecy key rate and characterized the optimal power distribution over coded layers. Our theoretical and numerical results have shown that the broadcast approach outperforms the single-level-coding based approach significantly, which establishes the important role of introducing self-interference in facilitating secret-key generation over slow-fading channels when transmit CSI is not available.

\appendices

\section{Proof of Theorem ~\ref{th:srate_cnt}}\label{layer-achievability}

Let us first consider the $L$-state fading wiretap channel defined by Definition \ref{def:L-state}. We have the following result.
\begin{lemma-A1}
For the $L$-state fading wiretap channel defined by Definition \ref{def:L-state}, the following key-rate is achievable:
\begin{equation}\label{Rs_Lstate}
R_s=\sum_{l_1}\sum_{l_2<l_1}\Pr\left(h_1=h^{[l_1]},h_2=h^{[l_2]}\right)\sum_{l=l_{2}+1}^{l_{1}}\left[r^{[l]} - \log\left(1+\frac{h^{[l_2]}p^{[l]}}{1+h^{[l_2]}\sum_{i=l+1}^{L}p^{[i]}}\right) \right],
\end{equation}
where we assume that $\{h^{[1]} \leq h^{[2]} \leq \dots \leq h^{[L]}\}$ and $r^{[l]}$ is given by
\begin{equation}\label{ll-rate2}
r^{[l]}=\log\left(1+\frac{h^{[l]}p^{[l]}}{1+h^{[l]}\sum_{i=l+1}^{L}p^{[i]}}\right).
\end{equation}
\end{lemma-A1}

\begin{proof}
We relegate the proof of Lemma A.1 to Appendix \ref{prf:lm-a1}.
\end{proof}

It is easy to observe that the result given by Theorem \ref{th:srate_cnt} is a continuous version of Lemma A.1 (as $L \rightarrow \infty$), and can be shown by following some standard steps in a straightforward manner. We omit these steps and next prove Lemma A.1 only.

\section{Proof of Lemma A.1}\label{prf:lm-a1}

\subsection{Secret-Key Generation For The $L$-State Fading Wiretap Channel}

The key-generation scheme for the $L$-state fading wiretap channel is similar to the scheme outlined in Section \ref{layer:scheme}. The encoding and decoding in the communication phase have been discussed in Section \ref{layer:layered:glc}. To proceed with the key generation phase, we will use the following notation (some of which has been explained previously but is repeated here for ease of reference).

Let $W_m = W_m^{[1:L]}$ represent the set of messages sent by Alice at the $m$-th time slot and $W_m^{[l]}$ represents the message sent at the $l$-th layer. At Bob, the reliably decoded message set at the $m$-th time slot is denoted by $W_m^{[\Dd_1]}$ and the undecodable message set is denoted by $W_m^{[\Du_1]}$.  At Eve, similarly, the reliably decoded message set is denoted by $W_m^{[\Dd_2]}$ and the undecodable message set is $W_m^{[\Du_2]}$.
We use $\Wv = (W_1, W_2, \dots, W_M)$ to represent the set of messages sent over all $M$ time slots. Similarly, $\Wv^{[\Dd_t]} = (W_1^{[\Dd_t]}, W_2^{[\Dd_t]}, \dots, W_M^{[\Dd_t]})$ and $\Wv^{[\Du_t]} = (W_1^{[\Du_t]}, W_2^{[\Du_t]}, \dots, W_M^{[\Du_t]})$ are defined for $t=1,2$.

We use $\Xv_m = \Xv_m^{[l:L]}$ to represent the set of codewords sent in the $m$-th time slot, $\Xv_m^{[\Dd_t]}$ and $\Xv_m^{[\Du_t]}$ (for $t=1,2$) to represent the sets of reliably decoded, and undecodable layers, respectively. Furthermore, $\Xv = (\Xv_1, \Xv_2, \dots, \Xv_M)$, $\Xv^{[\Dd_t]} = (\Xv_1^{[\Dd_t]}, \Xv_2^{[\Dd_t]}, \dots, \Xv_M^{[\Dd_t]})$ and $\Xv^{[\Du_t]} = (\Xv_1^{[\Du_t]}, \Xv_2^{[\Du_t]}, \dots, \Xv_M^{[\Du_t]})$ are the set, reliably decoded set, and undecodable set of codewords, respectively, over all $M$ time slots. In addition, $\Yv_1 = (\Yv_{1,1}, \Yv_{1,2}, \dots, \Yv_{1,M})$ and $\Yv_2 = (\Yv_{2,1}, \Yv_{2,2}, \dots, \Yv_{2,M})$ are the signals observed by Bob and Eve, respectively, over all $M$ time slots.

In the key generation phase, two parameters of the key generation codebook are $R$ and $R_s$. For the $L$-state fading wiretap channel, $R_s$ is given by (\ref{Rs_Lstate}) and $R$ is given by
\begin{equation}\label{R_Lstate}
R=\sum_{l=1}^{L}\Pr\left(h_1=h^{[l]}\right)\left(\sum_{i=1}^{l}r^{[i]}\right),
\end{equation}
where $r^{[i]}$ is given by (\ref{ll-rate2}).

\subsection{Genie-Leaked Information}

In the communication phase, we assume that the message conveyed by each layer is chosen independently of those at all other layers and uniformly at random. That is, at time slot $m$, the message $W_m^{[l]}$ sent by the $l$-th layer, is randomly and uniformly selected from $\{1, 2, \dots, 2^{Nr^{[l]}}\}$.  One can always assume that the random message is generated through a two-step procedure: first, two messages $\Wh_m^{[l]}$ and $\Wt_m^{[l]}$ are selected randomly and independently, where $\Wh_m^{[l]} \in \{1, \dots, 2^{N\rh_m^{[l]}}\}$ and $\Wt_m^{[l]} \in \{1, \dots, 2^{N\rt_m^{[l]})}\}$, where $\rt_m^{[l]}=r^{[l]}-\rh_m^{[l]}$; Then, message $W_m^{[l]}= \Wh_m^{[l]} \times \Wt_m^{[l]}$ is formed.

Note that this procedure is assumed only for facilitating the proof and is not actually required for encoding. In fact, $\rh_m^{[l]}$ can be any value as long as $0 \leq \rh_m^{[l]} \leq r^{[l]}$. For example, we can assume the following value for $\rh_m^{[l]}$:
\begin{equation}\label{ri1}
    \rh_m^{[l]} = \left\{ \begin{array}{ll}
 r^{[l]} &\mbox{ if $1 \leq l \leq l_{1m}$ (i.e., $l \in \Dd_{1m})$,} \\
 \min\left\{r^{[l]},\log\left(1+\frac{h_{2m}p^{[l]}}{1+h_{2m}\sum_{i=l+1}^{L}p^{[i]}}\right)\right\}  &\mbox{otherwise,}
       \end{array} \right.
\end{equation}
where $l_{1m}$ is the feedback layer index (i.e., the highest index of the decodable layers at Bob) in time slot $m$. Again, the feedback and channel information are not needed during the transmission since the two-step procedure is not actually executed.

Following the partitioning of messages, we have $\Wh_m^{[\Dd_1]} = W_m^{[\Dd_1]}$, $\Wt_m^{[\Dd_1]} = \emptyset$, and $W_m^{[\Du_1]}=\Wh_m^{[\Du_1]} \times \Wt_m^{[\Du_1]}$. Hence, $W_m$ is decomposed as $W_m = W_m^{[\Dd_1]} \times \Wh_m^{[\Du_1]} \times \Wt_m^{[\Du_1]}$. By letting $\Whv^{[\Du_1]} = (\Wh_1^{[\Du_1]}, \Wh_2^{[\Du_1]}, \dots, \Wh_M^{[\Du_1]})$ and $\Wtv^{[\Du_1]} = (\Wt_1^{[\Du_1]}, \Wt_2^{[\Du_1]}, \dots, \Wt_M^{[\Du_1]})$, we have $\Wv = \Wv^{[\Dd_1]} \times \Whv^{[\Du_1]} \times \Wtv^{[\Du_1]}$ correspondingly.

We assume that there is a genie who gives the message set $\Wtv^{[\Du_1]}$ to Eve. This is a useful step to enable us to give a bound on the equivocation rate with respect to the key $K$ at Eve.

One might wonder if this genie-leaked information benefits Eve and eventually reduces the achievable key rate. In Fig. \ref{fig:genieleak}, we illustrate that the genie-leaked information does not benefit Eve.
Here, let us consider a special $L$-state fading wiretap channel for which $L=3$ and the support of both Alice-Bob and Alice-Eve channel gains is $\{h^{[1]},h^{[2]}, h^{[3]}\}$. It is easy to see that $\Wt_m^{[\Du_1]} \neq \emptyset$ if and only if $h_{1m}=h^{[2]}$ and $h_{2m}=h^{[1]}$ for a time slot $m$. Therefore, we can focus on such a time slot.
We have $\Dd_{1m}= \{1,2\}$, $\Du_{1m}= \{3\}$, $\Dd_{2m}= \{1\}$, and $\Du_{2m}= \{2,3\}$.

$\Xv_m^{[1]}$ is decoded and subtracted by both Alice and Bob from their received signals. Therefore, we consider only $\Xv_m^{[2]}$ and $\Xv_m^{[3]}$, where $\Xv_m^{[2]}$ contributes to the key generation and Eve tries to deduce information on $\Xv_m^{[2]}$, while $\Xv_m^{[3]}$ plays the role of interference. Fig. \ref{fig:genieleak} shows the rate of information that Eve can deduce on $\Xv_m^{[2]}$ versus the rate of interference codebook. (The rate region resembles that of a multiple access channel. Some related discussion can be found in \cite{Tang:arXiv:09A}.)

Eve uses the genie-leaked information to reduce the rate of interference codebook. To achieve this, Eve uses  $\Wt_m^{[3]}$ to obtain a thinned codebook $\Cc^{[3]}(\Wt_m^{[3]})$. That is, among all the codewords in the original codebook $\Cc^{[3]}$, i.e. only the ones corresponding to $\Wt_m^{[3]}$ are kept and the rest are eliminated. However, if the side information is given properly, Eve does not benefit from the genie. As shown in Fig. \ref{fig:genieleak}, the side information does not help Eve's eavesdropping if
\begin{equation*}
    \rt_m^{[3]} \leq r^{[3]}-\log\left(1+h_{2m}p^{[3]}\right).
\end{equation*}
Under this condition, the pair of coding rates of $\Cc^{[2]}$ and $\Cc^{[3]}(\Wt_m^{[3]})$ is represented by any point on the line segment from A to B. A reward of
\begin{equation*}
    \Delta_m = r^{[2]} - \log\left(1+\frac{h_{2m}p^{[2]}}{1+h_{2m}p^{[3]}}\right)
\end{equation*}
is collected from time slot $m$ in contributing to the key generation.

\begin{figure}
  \centering
  \includegraphics[width=0.7\linewidth]{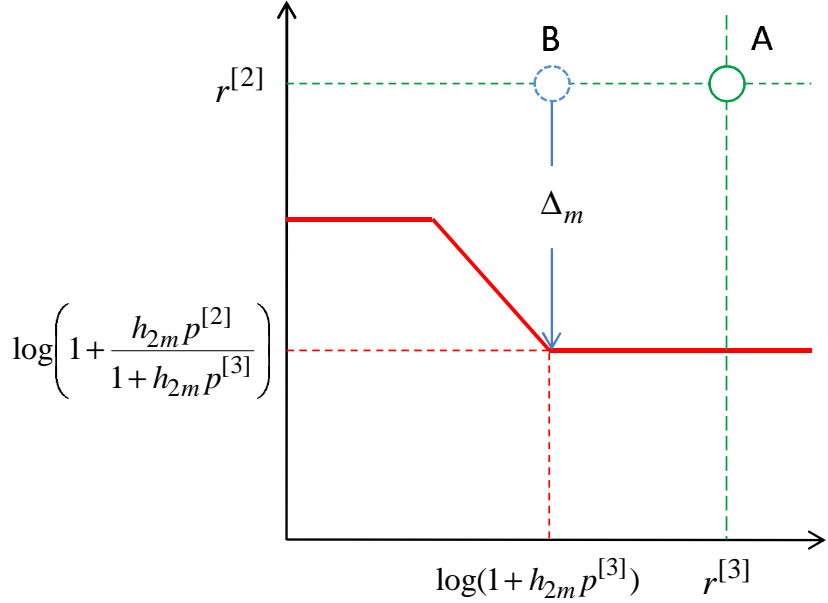}\\
  \caption{An illustrative example to show that the genie-leaked information does not benefit Eve.}\label{fig:genieleak}
\end{figure}

\subsection{Equivocation Calculation}

Now, we are ready to compute the equivocation rate with respect to the key $K$ at Eve:
\begin{align}
    &H(K|\Yv_2, \Psiv, \hv_2) \notag\\
    & \geq H(K|\Yv_2, \Psiv, \hv_1, \hv_2, \Wtv^{[\Du_1]} ) \label{erate_t1}\\
    & = H(K|\Yv_2, \Wtv^{[\Du_1]},\hv_1, \hv_2) \label{erate_t2}\\
    & = H(K, \Yv_2, \Xv|\Wtv^{[\Du_1]}, \hv_1, \hv_2) - H(\Yv_2|\Wtv^{[\Du_1]}, \hv_1, \hv_2)
     - H(\Xv|\Yv_2, K, \Wtv^{[\Du_1]}, \hv_1, \hv_2) \notag\\
    & \geq H(\Xv|\Yv_2, \Wtv^{[\Du_1]}, \hv_1, \hv_2)  - H(\Xv|\Yv_2, K, \Wtv^{[\Du_1]}, \hv_1, \hv_2), \label{erate_t3}
\end{align}
where (\ref{erate_t1}) is from the property that conditioning reduces entropy, (\ref{erate_t2}) is due to the fact that $\Psiv$ is a deterministic function of $\hv_1$ and $\Yv_2$.

As shown in Appendix \ref{pf_ET1} and \ref{pf_ET2}, the two terms in (\ref{erate_t3}) can be bounded as in the following,
\begin{equation}\label{ET1}
    H(\Xv|\Yv_2, \Wtv^{[\Du_1]}, \hv_1, \hv_2) \geq n(R_s-\delta_{N,M}),
\end{equation}
and
\begin{equation}\label{ET2}
    H(\Xv|\Yv_2, K, \Wtv^{[\Du_1]}, \hv_1, \hv_2) \leq n\delta'_{N,M},
\end{equation}
where $\delta_{N,M}, \delta'_{N,M} \rightarrow 0$ when $N, M \rightarrow \infty$.

By combining (\ref{erate_t3}), (\ref{ET1}) and (\ref{ET2}), we have
\begin{equation}\label{eqrate}
nR_e = H(K|\Yv_2, \Psiv, \hv_2) \geq n(R_s-\delta),
\end{equation}
which gives the perfect secrecy requirement that is
\begin{equation*}
R_e \geq R_s -\delta,
\end{equation*}
where $\delta \rightarrow 0$ as $n \rightarrow \infty$ (actually $N, M \rightarrow \infty$). Hence, we complete the proof.

\section{Proof of (\ref{ET1})}
\label{pf_ET1}

First, let us denote
\begin{equation*}
   E_1 \triangleq  H(\Xv|\Yv_2, \Wtv^{[\Du_1]}, \hv_1, \hv_2).
\end{equation*}
Due to independent coding at each time slot during forward transmission, we have
\begin{align}
    E_1 & = H(\Xv, \Yv_2, \Wtv^{[\Du_1]}, \hv_1, \hv_2) -  H(\Yv_2, \Wtv^{[\Du_1]}, \hv_1, \hv_2) \notag\\
    & = \sum_{m=1}^{M} H(\Xv_m, \Yv_{2m}, \Wt_m^{[\Du_1]}, h_{1m}, h_{2m}) - \sum_{m=1}^{M} H(\Yv_{2m}, \Wt_m^{[\Du_1]}, h_{1m}, h_{2m}) \notag\\
    & = \sum_{m=1}^{M} H(\Xv_m|\Yv_{2m}, \Wt_m^{[\Du_1]}, h_{1m}, h_{2m}). \notag
\end{align}
Furthermore, we have
\begin{align}
    E_1 &\geq \sum_{m \in \Mc^{+}}H(\Xv_m|\Yv_{2m}, \Wt_m^{[\Du_1]}, h_{1m}, h_{2m}) \label{erate_t5}\\
    & = \sum_{m \in \Mc^{+}} H(\Xv_m|\Wt_m^{[\Du_1]}, h_{1m}, h_{2m}) + H(\Yv_{2m} |\Xv_m, \Wt_m^{[\Du_1]}, h_{1m}, h_{2m}) - H(\Yv_{2m}|\Wt_m^{[\Du_1]}, h_{1m}, h_{2m})  \notag \\
    &= \sum_{m \in \Mc^{+}} H(\Xv_m|\Wt_m^{[\Du_1]}) + H(\Yv_{2m} |\Xv_m, h_{2m}) - H(\Yv_{2m}|\Wt_m^{[\Du_1]}, h_{1m}, h_{2m}) \label{erate_t5p1} \\
    &\geq \sum_{m \in \Mc^{+}} H(\Xv_m|\Wt_m^{[\Du_1]}) + H(\Yv_{2m} |\Xv_m, h_{2m}) - H(\Yv_{2m}|h_{2m})  \label{erate_t5p2} \\
    & \geq \sum_{m \in \Mc^{+}}H(\Xv_m|\Wt_m^{[\Du_1]}) -I(\Xv_m;\Yv_{2m}| h_{2m}) \label{erate_t6}
\end{align}
where $\Mc^{+}=\{m | m \in \{1, \dots, M \}, h_{1m} \geq h_{2m}\}$ is the set of time slots in which Alice-Bob channel is better than Alice-Eve channel), (\ref{erate_t5}) follows from the property that entropy is non-negative, (\ref{erate_t5p1}) follows from the property that $\Wt_m^{[\Du_1]} \leftrightarrow \Xv_m \leftrightarrow \Yv_{2m}$ forms a Markov chain, and (\ref{erate_t5p2}) follows from the property that conditioning reduces entropy.

To bound (\ref{erate_t6}) further, we have
\begin{align}
    &H(\Xv_m|\Wt_m^{[\Du_1]}) \notag\\
    &= N\left(\sum_{l=1}^{L}\rh_m^{[l]}\right) \notag\\
    &=N\left[\sum_{l=1}^{l_{1m}}r^{[l]}+\sum_{l=l_{1m}+1}^{L}\log\left(1+\frac{h_{2m}p^{[l]}}{1+h_{2m}\sum_{i=l+1}^{L}p^{[i]}}\right)\right]\label{erate_t8}\\ &=N\left[\sum_{l=1}^{l_{1m}}r^{[l]}+\log\left(1+h_{2m}\sum_{l=l_{1m}+1}^{L}p^{[l]}\right)\right]\label{erate_t9}
\end{align}
where $l_{1m}$ denotes the index of the highest decodable layer at Bob in time slot $m$, and (\ref{erate_t8}) follows from (\ref{ri1}). We also have
\begin{align}
    &I\left(\Xv_m;\Yv_{2m}| h_{2m}\right) \notag\\
    &= I\left(\Xv_m^{[\Dd_2]},\Xv_m^{[\Du_2]};\Yv_{2m}|h_{2m}\right) \notag \\
    &= I\left(\Xv_m^{[\Dd_2]};\Yv_{2m}|h_{2m}\right)+I\left(\Xv_m^{[\Du_2]};\Yv_{2m}|\Xv_m^{[\Dd_2]},h_{2m}\right) \notag\\
    &\leq H(\Xv_m^{[\Dd_2]})+ I\left(\Xv_m^{[\Du_2]};\Yv_{2m}|\Xv_m^{[\Dd_2]},h_{2m}\right) \notag\\
    &\leq N\left[\sum_{l=1}^{l_{2m}}r^{[l]} + \log\left(1+h_{2m}\sum_{l=l_{2m}+1}^{L}p^{[l]}\right)+\delta_1\right], \label{erate_t10}
\end{align}
where $l_{2m}$ denotes the index of the highest decodable layer at Eve in time slot $m$, and $\delta_1 \rightarrow 0$ as $N \rightarrow \infty$.

Combining (\ref{erate_t6}), (\ref{erate_t9}), and (\ref{erate_t10}), we have
\begin{align}
    E_1 &\geq N\left\{\sum_{m \in \Mc^{+}}\left[\sum_{l=l_{2m}+1}^{l_{1m}}r^{[l]} - \log\left(1+\frac{h_{2m}\sum_{l=l_{2m}+1}^{l_{1m}}p^{[l]}}{1+h_{2m}\sum_{l=l_{1m}+1}^{L}p^{[l]}}\right) - \delta_1 \right]\right\}\notag\\
    &= N \sum_{l_1}\sum_{l_2<l_1}\#\left(h_1=h^{[l_1]},h_2=h^{[l_2]}\right)\left[\sum_{l=l_{2}+1}^{l_{1}}r^{[l]} - \log\left(1+\frac{h^{[l_2]}\sum_{l=l_{2}+1}^{l_{1}}p^{[l]}}{1+h^{[l_2]}\sum_{l=l_{1}+1}^{L}p^{[l]}}\right) - \delta_1 \right]\notag\\
    &= N \sum_{l_1}\sum_{l_2<l_1}\#\left(h_1=h^{[l_1]},h_2=h^{[l_2]}\right)\left\{\sum_{l=l_{2}+1}^{l_{1}}\left[r^{[l]} - \log\left(1+\frac{h^{[l_2]}p^{[l]}}{1+h^{[l_2]}\sum_{i=l+1}^{L}p^{[i]}}\right) \right] - \delta_1 \right\},\notag
\end{align}
where $\#\left(h_1=h^{[l_1]},h_2=h^{[l_2]}\right)$ denotes the number of time slots (out of $M$ slots) that $h_1=h^{[l_1]}$ and $h_2=h^{[l_2]}$.

When $M \rightarrow \infty$, we have
\begin{align}
 E_1 &\geq N\sum_{l_1}\sum_{l_2<l_1}M\left[\Pr\left(h_1=h^{[l_1]},h_2=h^{[l_2]}\right)-\delta'_1\right]\left\{\sum_{l=l_{2}+1}^{l_{1}}\left[r^{[l]} - \log\left(1+\frac{h^{[l_2]}p^{[l]}}{1+h^{[l_2]}\sum_{i=l+1}^{L}p^{[i]}}\right) \right] - \delta_1\right\}\notag\\
 &= n(R_s-\delta_2),\label{erate_t11}
\end{align}
where $\delta_2 \rightarrow 0$ when $N \rightarrow \infty$ and $M \rightarrow \infty$.

\section{Proof of (\ref{ET2})}
\label{pf_ET2}

First, we denote
\begin{equation*}
E_2 \triangleq H(\Xv|\Yv_2, K, \Wtv^{[\Du_1]}, \hv_1, \hv_2).
\end{equation*}

To give a bound on $E_2$, we consider Eve's decoding of $\Xv$, i.e., the codewords sent over all $L$ layers and $M$ time slots, by assuming that Eve observes $\Yv_2$ and $\hv_2$, and is given (by a genie) the side information $K$, $\Wtv^{[\Du_1]}$ and $\hv_1$. Note that $\Xv = \Xv^{[\Dd_1]} \cup \Xv^{[\Du_1]}$, where $\Xv^{[\Du_1]}$ plays the role of interference and is not used in the key generation. To bound $E_2$, however, we need Eve to decode the interference given the genie-aided side information.

Given $\hv_1$ and $\hv_2$, Eve is able to partition $\Xv$ as
\begin{equation}\label{layerspartition}
\Xv=\Xv_{\Mc^{+}} \cup \Xv_{\Mc^{-}},
\end{equation}
where $\Mc^{+}=\{m | m =1, \dots, M, \mbox{~and~} h_{1m} \geq h_{2m}\}$, $\Mc^{-}=\{1,\dots,M\}/\Mc^{+}$, $\Xv_{\Mc^{+}}=\{\Xv_m | m \in \Mc^{+}\}$, and $\Xv_{\Mc^{-}}=\{\Xv_m | m \in \Mc^{-}\}$.
We consider the decoding of $\Xv_{\Mc^{+}}$ and $\Xv_{\Mc^{-}}$ separately as in the following subsections.

\subsection{Decoding of $\Xv_{\Mc^{-}}$}

We note that $\Xv_{\Mc^{-}}$ can be partitioned as
\begin{equation}\label{layerset2}
\Xv_{\Mc^{-}}=\Xv_{\Mc^{-}}^{[\Dd_2]} \cup \Xv_{\Mc^{-}}^{[\Du_2]}.
\end{equation}

Based on $\Yv_{2m}$ and side information $\Wt_m^{[\Du_1]}$, Eve performs the decoding of $\Xv_{\Mc^{-}}$ for each time slot $m \not\in \Mc^{+}$ independently. The decoding is performed in two steps:

\subsubsection{Decoding of $\Xv_{\Mc^{-}}^{[\Dd_2]}$}

For each $m \in \Mc^{-}$, Eve decodes $\Xv_m^{[\Dd_2]}$ (decodable layers for Eve) directly based on $\Yv_{2m}$ without using side information.

\subsubsection{Decoding of $\Xv_{\Mc^{-}}^{[\Du_2]}$}

After subtracting $\Xv_m^{[\Dd_2]}$ decoded previously, Eve attempts the decoding of $\Xv_m^{[\Du_2]}$ using the side information $\Wt_m^{[\Du_1]}$. More specifically, considering the decoding of $\Xv_m^{[l]}$ for layer $l \in \Du_{2m}$, we use $\Wt_m^{[l]}$, which is available since we have $\Du_{2m} \subset \Du_{1m}$ and therefore $\Wt_m^{[l]} \in \Wt_m^{[\Du_1]}$.
We denote by $\Cc^{[l]}(\Wt_m^{[l]})$ the thinned codebook corresponding to the genie-informed message $\Wt_m^{[l]}$. The size of $\Cc^{[l]}(\Wt_m^{[l]})$ is $2^{N\rh_{m}^{[l]}}$, where $\rh_{m}^{[l]}$ is given by (\ref{ri1}). Eve attempts to decode $\Xv_m^{[l]}$ using $\Cc^{[l]}(\Wt_m^{[l]})$ after subtracting the layers lower than $l$, denoted by $\Xv_m^{[1:(l-1)]}$.
For any typical sequences $\Xv_m^{[l]}$ and $\Yv_{2m}$, it can be shown that
\begin{equation*}
I\left(\Xv_m^{[l]};\Yv_{2m}|\Xv_m^{[1:(l-1)]}\right) \geq N \left[\log\left(1+\frac{h_{2m}p^{[l]}}{1+h_{2m}\sum_{i=l+1}^{L}p^{[i]}}\right)-\epsilon\right].
\end{equation*}
Hence, Eve is able to decode $\Xv_m^{[l]}$ with an arbitrarily small error probability when $N \rightarrow \infty$. By performing decoding for all $l \in \Du_{2m}$ successively, Eve decodes $\Xv_m^{[\Du_2]}$.

\subsection{Decoding of $\Xv_{\Mc^{+}}$}

We note that $\Xv_{\Mc^{+}}$ can be partitioned as
\begin{equation}\label{layerset1}
\Xv_{\Mc^{+}}=\Xv_{\Mc^{+}}^{[\Dd_2]} \cup \left(\Xv_{\Mc^{+}}^{[\Dd_1]} \cap \Xv_{\Mc^{+}}^{[\Du_2]}\right) \cup \Xv_{\Mc^{+}}^{[\Du_1]},
\end{equation}
and Eve performs the decoding of $\Xv_{\Mc^{+}}$ through the following three steps:

\subsubsection{Decoding of $\Xv_{\Mc^{+}}^{[\Dd_2]}$}

Eve decodes $\Xv_{\Mc^{+}}^{[\Dd_2]}$ directly based on $\Yv_2$ without using side information.

\subsubsection{Decoding of $\Xv_{\Mc^{+}}^{[\Dd_1]}\ \cap \Xv_{\Mc^{+}}^{[\Du_2]}$}

Eve decodes $\Xv_{\Mc^{+}}^{[\Dd_1]} \cap \Xv_{\Mc^{+}}^{[\Du_2]}$ jointly based on a list decoding argument, which is explained in details as in the following. A similar argument based on list decoding was given in \cite{Gopala:IT:08}.

\begin{figure*}
  \centering
  \includegraphics[width=0.7\linewidth]{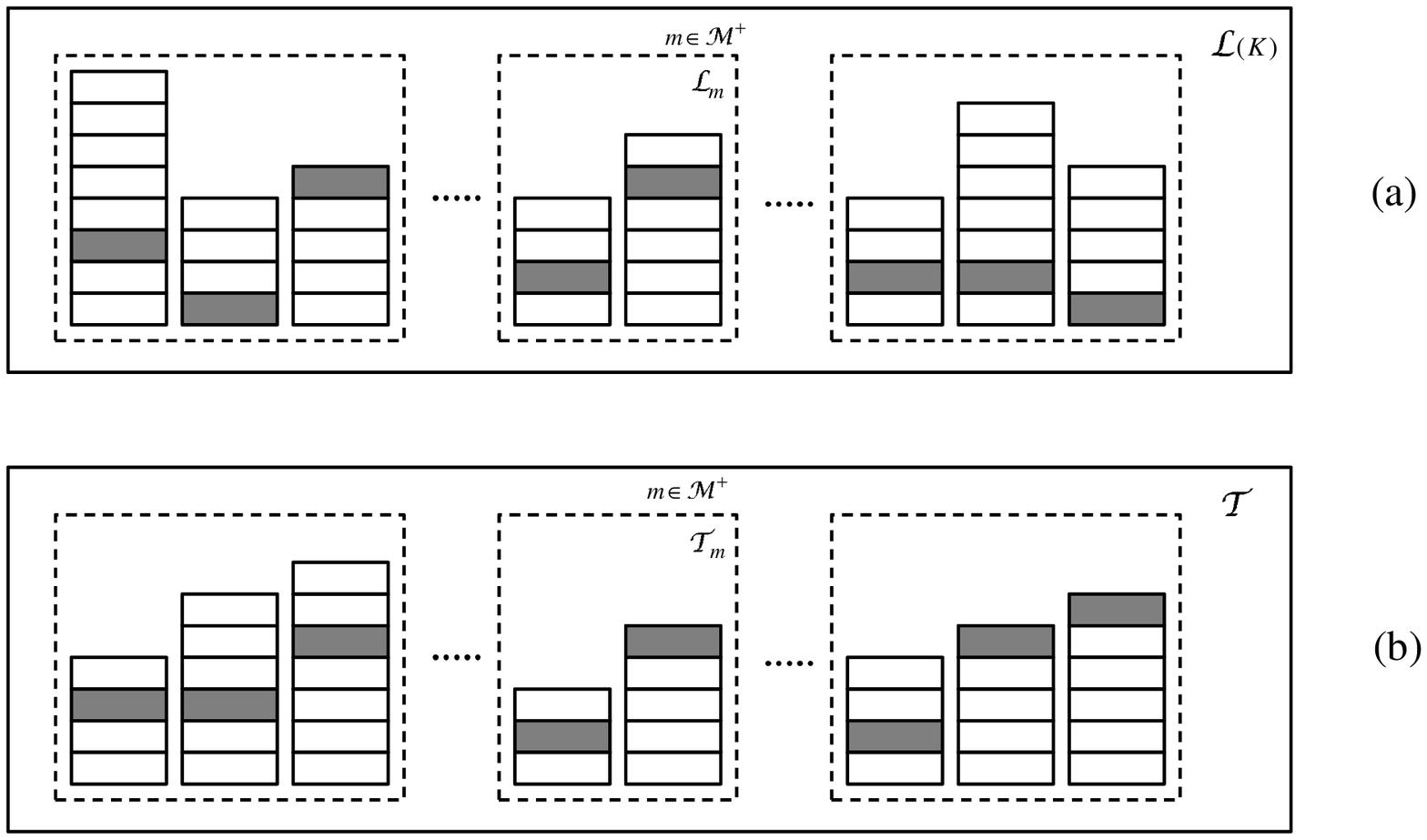}\\
  \caption{Two lists of super-sequences: (a) list $\Lc(K)$ constructed based on genie-provided $K$, ~(b) list $\Tc$ constructed based on joint-typicality.}\label{fig:list}
\end{figure*}

\begin{definition}
Sequence $\xlv_m$ is the concatenation of the codewords sent from the group of communication codebooks $\Cc^{[\Dd_{1m} \cap \Du_{2m}]}$ (i.e. $\xlv_m = \left[\xv_m^{[l_{2m}+1]}, \dots, \xv_m^{[l_{1m}]}\right]$). The concatenation of sequences $\xlv_m$ for all $m \in \Mc^{+}$ is called a super-sequence, denoted by $\xlv$.
\end{definition}

The length of sequence $\xlv_m$ is $N(l_{1m}-l_{2m})$, and the length of super-sequence $\xlv$ is therefore $N\sum_{m \in \Mc^{+}}(l_{1m}-l_{2m})$. Therefore, the length of a super-sequence depends on the channel realizations of $\hv_1$ and $\hv_2$ for a finite $M$. However, as $M \rightarrow \infty$, it can be seen that the length does not depend on the channel realizations.

As shown in Fig. \ref{fig:list}, Eve generates two lists of such super-sequences $\Lc$ and $\Tc$ based on genie-provided secret key $K$ and joint-typicality, respectively.

First, given a secret key $K$, Eve narrows down to bin $\mathcal{B}(K)$ in the key generation codebook. Since the mapping function $g$ is deterministic (one-to-one) and encoding in the communication phase is also deterministic, Eve is able to generate $\Lc(K)$, a list of super-sequences each of which corresponds to a codeword in bin $\mathcal{B}(K)$.
Hence, the size of $\Lc(K)$ is $\|\Lc(K)\|=2^{nR_s}$.

For each $m \in \Mc^{+}$ and any possible sequence $\xlv_m$, we define
\begin{equation*}
        \gamma(\xlv_m, \Yv_{2m}) = \left\{ \begin{array}{ll}
1 &\mbox{ if $(\Xv_m^{[\Dd_1 \cap \Du_2]} , \Yv_{2m})$ are jointly typical when $\Xv_{m}^{[\Dd_2]}$ are decoded and substraced from $\Yv_{2m}$.} \\
 0  &\mbox{otherwise}.
       \end{array} \right.
\end{equation*}
Eve constructs a list $\Lc_m$ such that
\begin{equation}\label{Lm}
    \Tc_m = \left\{\xlv_m  | \gamma(\xlv_m, \Yv_{2m}) = 1 \right\}.
\end{equation}
That is, $\Tc_m$ consists of the sequences such that the corresponding codewords coming from codebooks $\Cc^{[\Dd_{1m} \cap \Du_{2m}]}$ are jointly typical with $\Yv_{2m}$ given that $\Xv_{m}^{[\Dd_2]}$ has already been decoded and canceled.
Finally, Eve constructs a list $\Tc$ by concatenating sequences in $\Tc_m$ for all $m \in \Mc^{+}$.

Suppose that $\xlv$ is the super-sequence corresponding to the transmitted codewords $\Xv_{\Mc^{+}}^{[\Dd_1]}\ \cap \Xv_{\Mc^{+}}^{[\Du_2]}$. Given the two lists $\Lc(K)$ and $\Tc$, Eve attempts to find $\xlv$. Eve declares that $\xlv$ were sent, if $\xlv$ is the only common super-sequence in both $\Lc(K)$ and $\Tc$. She declares an error if there is no super-sequence or more than one super-sequences in $\Lc(K) \cap \Tc$.  Hence, there are two error events correspondingly,
\begin{enumerate}
  \item [$\Ee_{1}$]: $\xlv \not \in \Lc(K) \cap \Tc$,
  \item [$\Ee_{2}$]: there exists $\xlt \neq \xlv$, and $\xlt \in \Lc(K) \cap \Tc$.
\end{enumerate}
The Asymptotic Equipartition Property (AEP) implies that $\Pr(\Ee_{1}) \leq \epsilon_1$, where $\epsilon_1 \rightarrow 0$ as $n \rightarrow \infty$. $\Pr(\Ee_2)$ is bounded as the follows:
\begin{align}
    \Pr\left(\Ee_2\right) & \leq \E\left\{\sum_{\xlt \in \Tc, \xlt \neq \xlv}~\Pr\left(\xlt \in  \Lc(K)\right)\right\}\notag\\
    & \leq \E\left\{\|\Lc \|2^{-nR_s}\right\}, \label{evedecode_t1}
\end{align}
where $\|\Lc\|$ represents the size of the list $\Lc$, and (\ref{evedecode_t1}) follows from the uniform distribution of super-sequences in $\Lc(K)$.

To proceed, we need to give a bound on $\|\Lc\|$. We denote the size of $\Lc_m$ to be $\|\Lc_m\|$. For any $m \in \Mc^{+}$, $\|\Lc_m\|$ can be bounded as the follows:
\begin{align}
     \|\Lc_m\| & = \E\left\{\sum_{\xlv_m} \gamma(\xlv_m, \Yv_{2m})\right\} \notag\\
    & \leq 1+ \sum_{\xlt_{m} \neq \xlv_{m}}\E\left\{ \gamma(\xlv_m, \Yv_{2m})\right\}, \notag \\
    & \leq 1+ 2^{N\left(\sum_{l=l_{2m}+1}^{l_{1m}}r^{[l]}\right)}2^{N\left[-  \log\left(1+\frac{h_{2m}\sum_{l=l_{2m}+1}^{l_{1m}}p^{[l]}}{1+h_{2m}\sum_{l=l_{1m}+1}^{L}p^{[l]}}\right)+\epsilon_2\right]} \notag\\
   & \leq 2^{N\left\{\sum_{l=l_{2m}+1}^{l_{1m}}\left[r^{[l]}- \log\left(1+\frac{h_{2m}p^{[l]}}{1+h_{2m}\sum_{i=l+1}^{L}p^{[i]}}\right)\right]+ \epsilon_3\right\}},\notag
\end{align}
where $\epsilon_2, \epsilon_3 \rightarrow 0$ as $N \rightarrow \infty$.
The size of $\Lc$ is then bounded as
\begin{align}
    \|\Lc\| & = \prod_{m \in \Mc^{+}} \|\Lc_m\| \notag \\
    &\leq 2^{N\sum_{m \in \Mc^{+}}\left\{\sum_{l=l_{2m}+1}^{l_{1m}}\left[r^{[l]}- \log\left(1+\frac{h_{2m}p^{[l]}}{1+h_{2m}\sum_{i=l+1}^{L}p^{[i]}}\right)\right]+ \epsilon_3\right\}}.\notag
\end{align}
As $M \rightarrow \infty$, by following steps similar as those for deriving (\ref{erate_t11}), we have
\begin{align}
    \|\Lc\|  \leq 2^{n(R_s-\epsilon_4)}, \label{evedecode_t2}
\end{align}
Now we can combine (\ref{evedecode_t1}) and (\ref{evedecode_t2}) to obtain that
\begin{equation}\label{evedecode_t3}
 \Pr\left(\Ee_2\right) \leq 2^{-n\epsilon_4} \rightarrow 0,
\end{equation}
as $n \rightarrow \infty$. Hence, the average error probability for decoding $\Xv_{\Mc^{+}}^{[\Dd_1]} \cap \Xv_{\Mc^{+}}^{[\Du_2]}$ is bounded by
\begin{equation*}
    \Pr(\Ee_1 \cup \Ee_2) \leq \Pr(\Ee_1)+ \Pr(\Ee_2) \rightarrow 0,
\end{equation*}
as $n \rightarrow \infty$. Thus, Eve is able to find the right super-sequence $\xlv$ with a vanishing error probability. Since $\xlv$ and the group of codewords $\Xv_{\Mc^{+}}^{[\Dd_1]} \cap \Xv_{\Mc^{+}}^{[\Du_2]}$ are related by a one-to-one mapping, we conclude that Eve is able to decode $\Xv_{\Mc^{+}}^{[\Dd_1]} \cap \Xv_{\Mc^{+}}^{[\Du_2]}$ with a vanishing error probability.

\subsubsection{Decoding of $\Xv_{\Mc^{+}}^{[\Du_1]}$}

Eve subtracts $\Xv_{\Mc^{+}}^{[\Dd_2]}$ and $\Xv_{\Mc^{+}}^{[\Dd_1]} \cap \Xv_{\Mc^{+}}^{[\Du_2]}$ from $\Yv_2$ based on the two previous decoding procedures, and tries to decode $\Xv_{m}^{[\Du_1]}$ using the thinned codebooks $\Cc^{[\Du_1]}(\Wtv^{[\Du_1]})$. The decoding procedure is similar to that discussed in subsection A.2.

Finally, we conclude that Eve is able to decode $\Xv$ given $\Yv_2$, the genie-informed (secret-key) information $K$, and the side information $\Wtv^{[\Du_1]}$. Hence, Fano's inequality implies that
\begin{equation}\label{erate_term2}
   E_2= H(\Xv|\Yv_2, W, \Wtv^{[\Du_1]}, \hv_1, \hv_2) \leq n\delta_n \rightarrow 0,
\end{equation}
as $n \rightarrow \infty$.
We thus complete the proof of (\ref{ET2}).

\section{Proof of Lemma \ref{lm:srate2_cnt}}\label{srate2}

We can rewrite the secrecy key rate $R_s$ as
\begin{equation}\label{T1-T2}
R_s  = T_1 -T_2,
  \end{equation}
where
\begin{align}
T_1  &= \int_{0}^{\infty} \underbrace{\int_{0}^{h_1}\left[\int_{h_2}^{h_1}\frac{s\rho(s)ds}{1+sI(s)}\right]d[1-F_2(h_2)]}_{T_{1i}(h_1)}d[1-F_1(h_1)],\label{T1}\\
\mbox{and} \qquad T_2  &= \int_{0}^{\infty} \underbrace{\int_{0}^{h_1}\left[\int_{h_2}^{h_1}\frac{h_2\rho(s)ds}{1+h_2I(s)}\right]d[1-F_2(h_2)]}_{T_{2i}(h_1)}d[1-F_1(h_1)].\label{T2}
  \end{align}

\subsection{Evaluation of $T_1$}

The under-braced term $T_{1i}(h_1)$ can be evaluated by integrating by part. We have
\begin{align}\label{proof_T1i}
    T_{1i}(h_1) & = \left.\left[\int_{h_2}^{h_1}\frac{s\rho(s)ds}{1+sI(s)}\right]\left[1-F_2(h_{2})\right]\right|_{0}^{h_1}-\int_{0}^{h_1}\left[1-F_2(h_{2})\right]d\left[\int_{h_2}^{h_1}\frac{u\rho(u)du}{1+uI(u)}\right]\notag\\
    & = -\int_{0}^{h_1} \frac{s\rho(s)ds}{1+sI(s)}+\int_{0}^{h_1}\left[1-F_2(s)\right]\frac{s\rho(s)ds}{1+sI(s)}\notag\\
    &=-\int_{0}^{h_1}F_2(s)\frac{s\rho(s)ds}{1+sI(s)}
  \end{align}
By another integrating by part, we obtain
\begin{align}\label{t1_eval}
    T_1 & = \int_{0}^{\infty} T_{1i}(h_1)d\left[1-F_1(h_1)\right]\notag\\
        & = \left.T_{1i}(h_1)\left[1-F_1(h_1)\right]\right|_{0}^{\infty}-\int_{0}^{\infty}\left[1-F_1(h_1)\right]d\left[T_{1i}(h_1)\right]\notag\\
        &= -\int_{0}^{\infty}\left[1-F_1(h_1)\right]d\left[T_{1i}(h_1)\right]\notag\\
        &= \int_{0}^{\infty}\left[1-F_1(s)\right]F_2(s)\frac{s\rho(s)ds}{1+sI(s)}.
\end{align}

\subsection{Evaluation of $T_2$}

The under-braced term $T_{2i}(h_1)$ can be rewritten as
\begin{align*}
    T_{2i}(h_1) & = \left.\left[\int_{h_2}^{h_1}\frac{h_2\rho(s)ds}{1+h_2I(s)}\right]\left[1-F_2(h_{2})\right]\right|_{0}^{h_1}-\int_{0}^{h_1}\left[1-F_2(h_{2})\right]d\left[\int_{h_2}^{h_1}\frac{h_2\rho(s)ds}{1+h_2I(s)}\right]\notag\\
    & = -\int_{0}^{h_1} \left[1-F_2(s)\right]d\left[\int_{h_2}^{h_1}\frac{h_2\rho(s)ds}{1+h_2I(s)}\right].
  \end{align*}
Notice that
\begin{equation*}
    \int_{h_2}^{h_1}\frac{h_2\rho(s)ds}{1+h_2I(s)}=\log\left(1+h_2I(h_2)\right)-\log\left(1+h_2I(h_1)\right),
\end{equation*}
and therefore
\begin{equation*}
    \frac{d}{dh_2}\left[\int_{h_2}^{h_1}\frac{h_2\rho(s)ds}{1+h_2I(s)}\right]=\frac{I(h_2)-h_2 \rho(h_2)}{1+h_2I(h_2)}-\frac{I(h_1)}{1+h_2I(h_1)}.
\end{equation*}
Hence, $T_{2i}(h_1)$ can be written as
\begin{equation}\label{proof_T2i}
     T_{2i}(h_1) = -\int_{0}^{h_1}\left[1-F_2(h_2)\right]\left[\frac{I(h_2)-h_2 \rho(h_2)}{1+h_2I(h_2)}-\frac{I(h_1)}{1+h_2I(h_1)}\right]dh_2.
\end{equation}
Furthermore, we have
\begin{equation}\label{proof_T2i_dh1}
     \frac{d}{dh_1}T_{2i}(h_1) = -\left[1-F_2(h_1)\right]\left[\frac{I(h_1)-h_1\rho(h_1)}{1+h_1I(h_1)}\right]+\frac{d}{dh_1}\left[\int_{0}^{h_1}\left[1-F_2(h_2)\right]\frac{I(h_1)}{1+h_2I(h_1)}dh_2\right].
\end{equation}\
To proceed, we need to interchange the operation of differentiation with respect to $h_1$ with the operation of integration over $h_2$, where the integral domain is also a function of $h_1$. We use the property that for any real differentiable function $p(x,y)$, we can write
\begin{equation}\label{partialdiff}
\frac{d}{dx}\int_{0}^{x}p(x,y)dy=p(x,x)+\int_{0}^{x}\frac{\partial p(x,y)}{\partial x}dy.
\end{equation}
In particular, we have
\begin{align}\label{dterm21}
    &\frac{d}{dh_1}\left[\int_{0}^{h_1}\left[1-F_2(h_2)\right]\frac{I(h_1)}{1+h_2I(h_1)}dh_2\right] \notag\\
    &= \left[1-F_2(h_1)\right]\frac{I(h_1)}{1+h_1I(h_1)}+\int_{0}^{h_1}\left[1-F_2(h_2)\right]\frac{\partial}{\partial h_1}\left[\frac{I(h_1)}{1+h_2I(h_1)}\right]dh_2 \notag\\
    &= \left[1-F_2(h_1)\right]\frac{I(h_1)}{1+h_1I(h_1)}+\frac{\rho(h_1)}{I(h_1)}\int_{0}^{h_1}\left[1-F_2(h_2)\right]d\frac{1}{1+h_2I(h_1)} \notag\\
    &=\left[1-F_2(h_1)\right]\frac{I(h_1)}{1+h_1I(h_1)}+ \frac{\rho(h_1)}{I(h_1)}\left[\frac{1-F_2(h_1)}{1+h_1I(h_1)}-1+\int_{0}^{h_1}\frac{f_2(h_2)dh_2}{1+h_2I(h_1)}\right],
\end{align}
where we have used integrating by part to get to the last equality.

Putting (\ref{dterm21}) into (\ref{proof_T2i_dh1}), we have
\begin{equation}\label{proof_T2i_dh1_v2}
     \frac{d}{dh_1}T_{2i}(h_1) = -\frac{F_2(h_1)\rho(h_1)}{I(h_1)}+\frac{\rho(h_1)}{I(h_1)}\int_{0}^{h_1}\frac{f_2(h_2)dh_2}{1+h_2I(h_1)}.
\end{equation}

Now, we can evaluate $T_2$ by
\begin{align}\label{t2_eval}
    T_2 & = \int_{0}^{\infty} T_{2i}(h_1)d\left[1-F_1(h_1)\right]\notag\\
        & = -\int_{0}^{\infty} \left[1-F_1(h_1)\right] dT_{2i}(h_1)\notag \\
        & = \int_{0}^{\infty}\frac{\left[1-F_1(h_1)\right]F_2(h_1)\rho(h_1)}{I(h_1)}dh_1 - \int_{0}^{\infty}\frac{\left[1-F_1(h_1)\right]\rho(h_1)}{I(h_1)}\left[\int_{0}^{h_1}\frac{f_2(h_2)dh_2}{1+h_2I(h_1)}\right]dh_1.
\end{align}

\subsection{Evaluation of $R_s=T_1-T_2$}

Using (\ref{t1_eval}) and (\ref{t2_eval}), and replacing the variable $h_1$ and $h_2$ with $x$ and $y$, respectively, we have
\begin{align}\label{Rs_v1}
     R_s &= \int_{0}^{\infty}\frac{\left[1-F_1(x)\right]\rho(x)}{I(x)}\left[\int_{0}^{x}\frac{f_2(y)dy}{1+yI(x)}-\frac{F_2(x)}{1+xI(x)}\right] \notag\\
     &=\int_{0}^{\infty}\left[1-F_1(x)\right]\rho(x)\left[\int_{0}^{x}\frac{F_2(y)dy}{\left[1+yI(x)\right]^2}\right],
\end{align}
which is (\ref{srate2_cnt}).

\section{Proof of Lemma \ref{lm:Iterm_cnt}}
\label{Iterm}

The functional of (\ref{srate2_cnt}) is defined by
\begin{equation*}
    L\left(x,I(x),I'(x)\right)=-\left[1-F_1(x)\right]I'(x)\left[\int_{0}^{x}\frac{F_2(y)dy}{[1+yI(x)]^{2}}\right].
\end{equation*}
A necessary condition for a maximum of the integral of $L(x,I(x),I'(x))$ over $x$ is a zero variation of the  functional. For characterizing the optimal $I(x)$, the E\"{u}ler-Lagrangian equation \cite{Gelfand:63} gives a necessary condition denoted by
\begin{equation}\label{euler2}
    \frac{\partial L}{\partial I}-\frac{d}{dx}\left(\frac{\partial L}{\partial I'}\right)=0,
\end{equation}
for which we have,
\begin{align}
    \frac{\partial L}{\partial I} &= 2\left[1-F_1(x)\right]I'(x)\int_{0}^{x}\frac{yF_2(y)dy}{[1+yI(x)]^3},\label{euler_tm1}\\
    \frac{\partial L}{\partial I'} &= -\left[1-F_1(x)\right]\int_{0}^{x}\frac{F_2(y)dy}{[1+yI(x)]^2},\label{euler_tm2}\\
    \frac{d}{dx}\frac{\partial L}{\partial I'} &= f_1(x)\int_{0}^{x}\frac{F_2(y)dy}{\left[1+yI(x)\right]^2}-\left[1-F_1(x)\right]\frac{d}{dx}\int_{0}^{x}\frac{F_2(y)dy}{[1+yI(x)]^2},\label{euler_tm3}
\end{align}
with
\begin{equation}\label{euler_tm4}
    \frac{d}{dx}\int_{0}^{x}\frac{F_2(y)dy}{[1+yI(x)]^2} = \frac{F_2(x)}{\left[1+xI(x)\right]^2}-2I'(x)\int_{0}^{x}\frac{yF_2(y)dy}{\left[1+yI(x)\right]^3}.
\end{equation}
Using (\ref{euler_tm1}), (\ref{euler_tm3}), and (\ref{euler_tm4}) in (\ref{euler2}), we have
\begin{equation}\label{condition}
    \int_{0}^{x}\frac{F_2(y)dy}{[1+yI(x)]^2} = \frac{\left[1-F_1(x)\right]F_2(x)}{\left[1+xI(x)\right]^2f_1(x)}.
\end{equation}
Hence, we proved Lemma \ref{lm:Iterm_cnt}.

\section{Proof of (\ref{rate_l1})}\label{proof_onelevel_term1}

According to Lemma \ref{single-level-lemma}, the secrecy rate is
\begin{align}
   R_s^{[1]}&=\Pr\left[R^{[1]} \leq \log(1+h_1P)\right]\E_{h_2}\left[R^{[1]}- \log\left(1+h_2P\right)\right]^{+}\notag\\
   &=\Pr\left\{h_1 \geq h_1^{*}\right\}\int_{0}^{h_1^{*}}\left[R^{[1]}- \log\left(1+h_2P\right)\right]f_2(h_2)dh_2 \notag\\
   &=\exp\left(-\frac{h_1^{*}}{\lambda_1}\right)\left[R^{[1]}F_2(h_1^{*})-\int_{0}^{h_1^{*}}\log(1+h_2P)f_2(h_2)dh_2\right], \label{single_int_pat1}
\end{align}
where $h_1^{*}=\left[\exp(R^{[1]})-1\right]/{P}$. By using integrating by part for the integral in (\ref{single_int_pat1}), we have
\begin{align*}
   R_s^{[1]}&=\exp\left(-\frac{h_1^{*}}{\lambda_1}\right)\int_{0}^{h_1^{*}}\frac{\left[1-\exp\left(-h_2/\lambda_2\right)\right]P}{1+h_2P}dh_2\notag\\
   &=\exp\left(-\frac{h_1^{*}}{\lambda_1}\right)\left[R^{[1]}-\int_{0}^{h_1^{*}}\frac{\exp\left(-h_2/\lambda_2\right)P}{1+h_2P}dh_2\right].\notag
\end{align*}
By letting $t=(1+h_2P)/(\lambda_2 P)$, we have
\begin{align*}
   R_s^{[1]}  &=\exp\left(-\frac{h_1^{*}}{\lambda_1}\right)\left[R^{[1]}-\exp\left(\frac{1}{\lambda_2 P}\right)\int_{\frac{1}{\lambda_2 P}}^{\frac{1+h_1^{*} P}{\lambda_2 P}}\frac{\exp(-t)}{t}dt\right]\notag\\
   &=\exp\left(-\frac{h_1^{*}}{\lambda_1}\right)\left\{R^{[1]}-\exp\left(\frac{1}{\lambda_2 P}\right)\left[E_i\left(\frac{1+h_1^* P}{\lambda_2 P}\right)-E_i\left(\frac{1}{\lambda_2 P}\right)\right]\right\}\notag
\end{align*}
By using $h_1^{*}=\left[\exp(R^{[1]})-1\right]/{P}$, we can obtain (\ref{rate_l1}).

\section{Proof of (\ref{int_term_ap})}\label{proof_int_term1}

We can write
\begin{align}\label{int-termapp}
\int_{0}^{x}\frac{F_2(y)dy}{\left[1+yI(x)\right]^2}= \int_{0}^{x}\frac{1-\exp\left(-y/\lambda_2\right)}{\left[1+yI(x)\right]^2}dy
=\underbrace{\int_{0}^{x}\frac{dy}{\left[1+yI(x)\right]^2}}_{T_3}-\underbrace{\int_{0}^{x}\frac{\exp\left(-y/\lambda_2\right)dy}{\left[1+yI(x)\right]^2}}_{T_4},
\end{align}
and evaluate $T_3$ and $T_4$ separately. First, we have
\begin{equation}\label{T3-1}
    T3=\frac{x}{1+xI(x)}.
\end{equation}
To evaluate $T_4$, we have
\begin{align}\label{T4-1}
T_4&=-\int_{0}^{x}\frac{\exp\left(-y/\lambda_2\right)}{I(x)}d\left[\frac{1}{1+yI(x)}\right]\notag\\
&=-\left.\frac{1}{I(x)}\frac{\exp\left(-y/\lambda_2\right)}{1+yI(x)}\right|_{0}^{x}-\frac{1}{\lambda_2I(x)}\int_{0}^{x}\frac{\exp\left(-y/\lambda_2\right)dy}{1+yI(x)}\notag\\
&=\frac{1}{I(x)}\left[1-\frac{\exp\left(-x/\lambda_2\right)}{1+xI(x)}\right]-\underbrace{\frac{1}{\lambda_2I^2(x)}\int_{0}^{x}\frac{\exp\left(-y/\lambda_2\right)}{1+yI(x)}d\left[1+yI(x)\right]}_{T_5}.
\end{align}
By letting $1+yI(x)=t$, we have
\begin{align}\label{T4-2}
T_5&=\frac{\exp\left(1/\lambda_2I(x)\right)}{\lambda_2I^2(x)}\int_{1}^{1+xI(x)}\frac{\exp\left(-t/\lambda_2I(x)\right)}{t/\lambda_2I(x)}d\left[\frac{t}{\lambda_2I(x)}\right]\notag\\
&=\frac{1}{\lambda_2I^2(x)}\exp\left(\frac{1}{\lambda_2I(x)}\right)E_i\left.\left(\frac{t}{\lambda_2I(x)}\right)\right|_{1+xI(x)}^{1}\notag\\
&=\frac{1}{\lambda_2I^2(x)}\exp\left(\frac{1}{\lambda_2I(x)}\right)\left[E_i\left(\frac{1}{\lambda_2I(x)}\right)-E_i\left(\frac{1+xI(x)}{\lambda_2I(x)}\right)\right].
\end{align}
Combining (\ref{T3-1}), (\ref{T4-1}), (\ref{T4-2}) and (\ref{int-termapp}), we can obtain (\ref{int_term_ap}).

\bibliographystyle{IEEEtran}

\end{document}